\documentclass[pdflatex,12pt]{amsart}
\usepackage{pdfsync}

\pdfoutput=1  %%%% Forces arXive.org to compile with PDF LaTeX

\usepackage{amsmath,amsfonts}
\usepackage{eucal}
\usepackage{enumerate}
\usepackage{xy}
\xyoption{matrix} \xyoption{arrow} \xyoption{arc} \xyoption{color}
\xyoption{pdf}

\makeatletter
\newtheoremstyle{corsivo}
   {\medskipamount}{\medskipamount}%
   {\itshape}{}%
   {\bfseries}{}%
   { }
   {\thmname{#1}\thmnumber{\@ifnotempty{#1}{ }\@upn{#2}}%
    \thmnote{ {\bfseries(#3)}}.}%
\makeatother

\theoremstyle{corsivo}
\newtheorem{theorem}{Theorem}[section]
\newtheorem{lemma}[theorem]{Lemma}
\newtheorem{crl}[theorem]{Corollary}
\newtheorem{prop}[theorem]{Proposition}
\newtheorem{df}[theorem]{Definition}

\makeatletter
\newtheoremstyle{dritto}
   {\medskipamount}{\medskipamount}%
   {\rmfamily}{}%
   {\bfseries}{}%
   { }
   {\thmname{#1}\thmnumber{\@ifnotempty{#1}{ }\@upn{#2}}%
    \thmnote{ {\bfseries(#3)}}.}%
\makeatother

\theoremstyle{dritto}
\newtheorem{rmk}[theorem]{Remark}
\newtheorem{datum}[theorem]{Datum}
\newtheorem{assumption}[theorem]{Assumption}

\newcommand{\ph}{\theta_q}
\newcommand{\B}{\mathbb{B}}
\newcommand{\C}{\mathbb{C}}
\newcommand{\R}{\mathbb{R}}
\newcommand{\Z}{\mathbb{Z}}
\newcommand{\N}{\mathbb{N}}
\newcommand{\T}{\mathbb{T}}
\newcommand{\calL}{\mathcal{L}}
\newcommand{\PB}{\mathcal{P}}
\newcommand{\ch}{\mathrm{ch}}
\newcommand{\BF}{\mathcal{B}}
\newcommand{\Hi}{\mathcal{H}}
\newcommand{\scal}[2]{\left\langle #1, #2 \right\rangle}
\newcommand{\eu}{\mathrm{e}}
\newcommand{\iu}{\mathrm{i}}
\newcommand{\di}{\mathrm{d}}
\newcommand{\sub}[1]{_{\mathrm{#1}}}
\newcommand{\nv}{n\sub{v}}

\renewcommand{\1}{\mathbf{1}} %%%%%%%%%%{\mathds{1}}
\renewcommand{\dot}[1]{\overset{\circ}{#1}}

\DeclareMathOperator{\tr}{Tr}
\DeclareMathOperator{\Ran}{Ran}

\newcommand{\Base}{\T_d^*}

\newcommand{\Hf}{\mathcal{H}_{\mathrm{f}}}
\newcommand{\In}{{I}}
\newcommand{\Or}{\mathcal{O}}
\newcommand{\de}{\partial}
\newcommand{\Bi}{\mathcal{B}}

\newcommand{\virg}[1]{``#1''}
\newcommand{\ie}{{\sl i.\,e.\ }}
\newcommand{\eg}{{\sl e.\,g.\ }}
\newcommand{\expo}[1]{\mathrm{e}^{#1}}     % Exponential function
\newcommand{\inner}[2]{\left \langle  #1 \, , \,  #2 \right \rangle}
\newcommand{\norm}[1]{\| #1 \|}
\newcommand{\bra}[1]{\left\langle #1 \right|}
\newcommand{\ket}[1]{\left| #1 \right\rangle}

\newcommand{\set}[1]{ \left\{  #1 \right\}}

\newcommand{\half}{1/2}

\let\oldfootnote\footnote
\renewcommand{\footnote}[1]{\oldfootnote{\  #1}}

\usepackage{color}

\newcommand{\blu}{\textcolor[rgb]{0.05,0.24,0.57}}
\newcommand{\Composite}{stratified }
\newcommand{\vorticity}{eigenspace vorticity }

\title[Wannier functions in graphene]{Topological invariants  \\[1mm] of eigenvalue intersections \\[1mm] and decrease of Wannier functions  in graphene}
\author{Domenico Monaco \and Gianluca Panati}
\date{January 12, 2014}

\usepackage[pdftex,
]{hyperref}

\begin{document}

\begin{abstract}
We investigate the asymptotic decrease of the Wannier functions for the valence and conduction band of graphene, both in the monolayer and the multilayer case. Since the decrease of the Wannier functions is characterised by the structure of the Bloch eigenspaces around the Dirac points, we introduce a geometric invariant of the family of eigenspaces, baptised {eigenspace vorticity}. We compare it with the pseudospin winding number. For every value $n \in \Z$ of the eigenspace vorticity, we exhibit a canonical model for the local topology of the eigenspaces. With the help of these canonical models, we show that the single band Wannier function $w$  satisfies $|w(x)| \le \mathrm{const} \cdot |x|^{-2}$ as $|x| \to \infty$, both in monolayer and bilayer graphene.

\bigskip

\noindent \textsc{Keywords}: Wannier functions, Bloch bundles, conical intersections, eigenspace vorticity, pseudospin winding number, graphene.
\end{abstract}

\maketitle

\begin{center}
\textsl{Dedicated to Herbert Spohn, with admiration}
%\textsl{on the occasion of his 66$^{th}$ birthday}
\end{center}

\bigskip

\tableofcontents

%%%%%%%% INTRODUCTION %%%%%%%%%%%%%%%%%%%%%%%%%%%%%%%%%%%%%%%%%%%%%%%%%%%%%%%%%%

\section*{Introduction} \label{sec:intro}

The relation between topological invariants of the Hamiltonian and localization and transport properties of the electrons has become, after a profound paper by Thouless \emph{et al.}\, \cite{TKNN}, a paradigm of theoretical and mathematical physics. Besides the well-known example of the Quantum Hall effect \cite{Bellissard94,Graf review}, the same paradigm applies to the macroscopic polarization of  insulators under time-periodic deformations \cite{KSV,Re92,PaSparTe} and to many other examples \cite{XiaoChangNiu}. While this relation has been deeply investigated in the case of gapped insulators, the case of semimetals remains, to our knowledge, widely unexplored. In this paper, we consider the prototypical example of graphene \cite{CastroNeto, Goerbig, BenaMontambaux}, both in the monolayer and in the multilayer realisations, and we investigate the relation between a \emph{local} geometric invariant of the eigenvalue intersections and the electron localization.

A fundamental tool to study the localization of the electrons in periodic and almost-periodic systems is provided by Wannier functions \cite{Wannier,Wannier review}. In the case of a single Bloch band isolated from the rest of the spectrum, the existence of an exponentially localized Wannier function was proved in dimension $d=1$ by W. Kohn for centrosymmetric crystals \cite{Kohn64}. The latter hypothesis has been later removed by J. de Cloizeaux \cite{Cl2}. A proof of existence for $d \leq 3$ has been obtained by J. de Cloizeaux for centrosymmetric crystals \cite{Cl1, Cl2}, and by G. Nenciu \cite{Ne83} in the general case.

Whenever the Bloch bands intersect each other, there are two possible approaches. On the one hand, following de Cloizeaux \cite{Cl2}, one considers a relevant family of Bloch bands which are separated by a gap from the rest of the spectrum (\eg the bands below the Fermi energy in an insulator). Then the notion of Bloch function is relaxed to the weaker notion of \emph{quasi-Bloch function}, and one investigates whether the corresponding \emph{composite} Wannier functions are exponentially localized. An affirmative answer was provided by G. Nenciu for $d=1$ \cite{Ne91}, and only recently for $d \leq 3$ \cite{Wannier_letter_BPCM,Pa2008}. On the other hand, one may focus on a \emph{single} non-isolated band and estimate the asymptotic decrease of the corresponding single-band Wannier function, as $|x| \to \infty$. The rate of decrease depends, roughly speaking, on the regularity of the Bloch function at the intersection points.

In this paper we follow the second approach.  We consider the case of graphene (both monolayer and bilayer) \cite{CastroNeto, Goerbig} and we explicitly compute the rate of decrease of the Wannier functions corresponding to the conduction and valence band. Since the rate of decrease crucially depends on the behaviour of the Bloch functions at the Dirac points, we preliminarily study the topology of the Bloch eigenspaces around those points.

More precisely, we introduce a geometric invariant of the eigenvalue intersection, which encodes the behaviour of the Bloch eigenspaces at the singular point (Section \ref{sec:vorticity}). We show that  our invariant, baptised \emph{eigenspace vorticity}, equals the \emph{pseudospin winding number} \cite{Park-Marzari, Novoselov_McCann_et_al2006, McCann_Falko2006} whenever the latter is well-defined (Section \ref{sec:pseudospin}). We prove, under suitable assumptions,  that if the value of the \vorticity is $\nv \in \Z$, then the local behaviour of the Bloch eigenspaces is described by the $\nv$-canonical model, explicitly described in Section \ref{sec:models}.  For example, monolayer and bilayer graphene correspond to the cases $\nv=1$ and $\nv =2$, respectively.    The core of our topological analysis is Theorem \ref{Th:criterion}, which shows that, in the relevant situations, the family of canonical models provides a complete classification of the local behaviour of the eigenspaces.

As a consequence of the previous geometric analysis, in Section \ref{sec:decrease} we are able to compute the rate of decrease of Wannier functions corresponding to the valence and conduction bands of monolayer and bilayer graphene. For both bands, we essentially obtain that
\begin{equation} \label{FreeElectronGas}
|w(x)| \le \mathrm{const} \cdot |x|^{-2} \qquad \mbox{as } |x| \to \infty,
\end{equation}
see Theorems \ref{Asymp_wcan} and \ref{Asymp_graphene} for precise statements.

The power-law decay in \eqref{FreeElectronGas} suggests that electrons in the conduction or valence band are delocalized. The absence of localization and the finite metallic conductivity in monolayer graphene are usually explained as a consequence of the Dirac-like (conical) energy spectrum. Since bilayer graphene has the usual parabolic spectrum, \virg{\emph{the observation of the maximum resistivity $\approx h/4e^2$ [...] is most unexpected}}\, \cite{Novoselov_McCann_et_al2006}, thus challenging theoreticians to provide an explanation of the absence of localization in bilayer graphene. \newline In this paper, we show that the absence of localization is not a direct consequence of the conical spectrum, but it is rather a consequence of the non-smoothness of the Bloch functions at the intersection points, which is in turn a consequence of a non-zero eigenspace vorticity, a condition which is verified by both mono- and multilayer graphene.  While the existence of a local geometric invariant distinguishing monolayer graphene from bilayer graphene has been foreseen by several authors, as \eg \cite{Novoselov_McCann_et_al2006, McCann_Falko2006, Park-Marzari}, our paper first demonstrates the relation between non-trivial local topology and absence of localization in position space.

While the Wannier functions in graphene are the  motivating example, the importance of our topological analysis goes far beyond the specific case:  we see a wide variety of possible applications, ranging from the topological phase transition in the Haldane model \cite{Haldane88},  to the analysis of the conical intersections arising in systems of ultracold atoms in optical lattices \cite{ZhuWangDuan, Lepori11, Tarruel_et_al2012}, to a deeper understanding of the invariants in $3$-dimensional topological insulators \cite{HasanKane, KaneMele2005, SoluyanovVanderbilt2010}. As for the latter item, the applicability of our results is better understood in terms of edge states, following \cite[Sec. IV]{HasanKane}. Indeed, in a $3d$ crystal occupying the half-space, the edge states are decomposed with respect to a $2d$ crystal momentum; on the corresponding $2d$ Brillouin zone, there are four points invariant under time-reversal symmetry where surface Bloch bands may be doubly degenerate, yielding an intersection of eigenvalues. Although a detailed analysis is postponed to future work, we are confident that methods and techniques developed in this paper will contribute to a deeper understanding of the invariants of topological insulators.

\medskip

\textbf{Acknowledgments.} We are indebted with D. Fiorenza and  A. Pisante for many inspiring discussions, and with R. Bianco,
R. Resta and A. Trombettoni for interesting comments and remarks. We are also grateful to the anonymous reviewers for their useful observations and suggestions. Financial support from the INdAM-GNFM project \virg{Giovane Ricercatore 2011}, and from the AST Project 2009 \virg{Wannier functions} is gratefully
acknowledged.

%%%%%%%% SECTION 1 %%%%%%%%%%%%%%%%%%%%%%%%%%%%%%%%%%%%%%%%%%%%%%%%%%%%%%%%%%

\section{Basic concepts} \label{sec:basics}

In this Section, we briefly introduce the basic concepts and the notation, referring to \cite[Section 2]{Panati Pisante} for details.

\setcounter{subsection}{-1}

%%%%%%%% SECTION 1.0 %%%%%%%%%%%%%%%%%%%%%%%%%%%%%%%%%%%%%%%%%%%%%%%%%%%%%%%%%%

\subsection{Bloch Hamiltonians} \label{sec:BlochHamiltonians}

To motivate our definition, we initially consider a periodic Hamiltonian $H_{\Gamma} = - \Delta + V_{\Gamma} $ where $V_{\Gamma}(x + \gamma) = V_{\Gamma}(x)$ for every $\gamma$ in the periodicity lattice $ \Gamma = \mathrm{Span}_{\Z}\set{a_1, \ldots, a_d}$ (here $\set{a_1, \ldots, a_d}$ is a linear basis of $\R^d$). We assume that the potential $V_\Gamma$ defines an operator which is relatively bounded with respect to $\Delta$ with relative bound zero, in order to guarantee that $H_\Gamma$ is self-adjoint on the domain $W^{2,2}(\R^d)$: when $d=2$ (the relevant dimension for our subsequent analysis), for example, this holds whenever $V_\Gamma \in L^2\sub{loc}(\R^2)$ \cite[Thm. XIII.96]{Reed-Simon}. To study such \emph{Bloch Hamiltonians}, one introduces the (\emph{modified}) \emph{Bloch-Floquet transform} $\mathcal{U}\sub{BF}$, acting on a function $w \in \mathcal{S}(\R^d)$ as
\[ \left( \mathcal{U}\sub{BF} w \right)(k,y) := \frac{1}{|\B|^{\half}} \sum_{\gamma \in \Gamma} \expo{-\iu k \cdot (y+\gamma)} w(y+\gamma), \quad k \in \R^d, \: y \in \R^d. \]
Here $\B$ is the fundamental unit cell for the dual lattice $\Gamma^*$, \ie the lattice generated over the integers by the dual basis $\set{a_1^*, \ldots, a_d^*}$ defined by the relations $a_i^* \cdot a_j = 2 \pi \delta_{i,j}$. As can be readily verified, the function $\mathcal{U}\sub{BF} w$ is $\Gamma^*$-pseudoperiodic and $\Gamma$-periodic, meaning that
\begin{align*}
\left( \mathcal{U}\sub{BF} w \right)(k+\lambda,y) & = \expo{-\iu \lambda \cdot y} \, \left( \mathcal{U}\sub{BF} w \right)(k,y) && \text{for } \lambda \in \Gamma^*, \\
\left( \mathcal{U}\sub{BF} w \right)(k,y+\gamma) & = \left( \mathcal{U}\sub{BF} w \right)(k,y) && \text{for } \gamma \in \Gamma.
\end{align*}
Consequently, the function $\left( \mathcal{U}\sub{BF} w \right)(k,\cdot)$, for fixed $k \in \R^d$, can be interpreted as an element of the $k$-independent Hilbert space $\Hf = L^2(\T^d_{Y})$, where $\T^d_{Y} = \R^d/\Gamma$ is a $d$-dimensional torus in position space.

The Bloch-Floquet transform $\mathcal{U}\sub{BF}$, as defined above, extends to a unitary operator
\[ \mathcal{U}\sub{BF} \colon L^2(\R^d) \to \int_{\B}^{\oplus} \di k \, \Hf \]
whose inverse is given by
\[ \left( \mathcal{U}\sub{BF}^{-1} u \right)(x) = \frac{1}{|\B|^{\half}} \int_{\B} \di k \, \expo{\iu k \cdot x} u(k,[x]), \quad x \in \R^d, \]
where $[x] := x \bmod \Gamma$. With the hypotheses on the potential $V_\Gamma$ mentioned above, one verifies that $H_\Gamma$ becomes a fibred operator in Bloch-Floquet representation, namely
\begin{equation} \label{H(k)}
\mathcal{U}\sub{BF} H_\Gamma \mathcal{U}\sub{BF}^{-1} = \int_{\B}^{\oplus} \di k \, H(k) \quad \text{where} \quad H(k) = (- \iu \nabla_y + k)^2 + V_\Gamma.
\end{equation}
Each $H(k)$ acts on the $k$-independent domain $\mathcal{D} := W^{2,2}(\T^d_Y) \subset \Hf$, where it defines a self-adjoint operator. Since for any $\kappa_0 \in \C^d$ one has
\[ H(\kappa) = H(\kappa_0) + 2 (\kappa - \kappa_0) \cdot (- \iu \nabla_y) + (\kappa^2 - \kappa_0^2) \1 \]
and $(- \iu \nabla_y)$ is relatively bounded with respect to $H(\kappa_0)$, the assignment $\kappa \mapsto H(\kappa)$ defines an entire family of type (A), and hence an entire analytic family in the sense of Kato \cite[Thm. XII.9]{Reed-Simon}.

Moreover, all the operators $H(k)$, $k \in \R^d$, have compact resolvent, and consequently they only have pure point spectrum accumulating at infinity. We label the eigenvalues in increasing order, \ie $E_0(k) \le \cdots \le E_n(k) \le E_{n+1}(k) \le \cdots$, repeated according to their multiplicity; the function $k \mapsto E_n(k)$ is usually called the \emph{$n$-th Bloch band}. We denote by $u_n(k)$ the solution to the eigenvalue problem
\begin{equation} \label{BlochFunction}
 H(k) u_n(k) = E_n(k) u_n(k), \quad u_n(k, \cdot) \in \mathcal{D} \subset \Hf.
\end{equation}
The function $k \mapsto \psi_n(k, y) = \expo{\iu k \cdot y} u_n(k,y)$ is called the \emph{$n$-th Bloch function} in the physics literature; $u_n(k, \cdot)$ is its $\Gamma$-periodic part.

%%%%%%%% SECTION 1.1 %%%%%%%%%%%%%%%%%%%%%%%%%%%%%%%%%%%%%%%%%%%%%%%%%%%%%%%%%%

\subsection{From insulators to semimetals} \label{sec:ins->semi}

In the case of an isolated Bloch band, or an isolated family of $m$ Bloch bands $\set{E_n, \ldots, E_{n+m-1}} = \set{E_i}_{i \in \In}$, where \virg{isolated} means that
\begin{equation}\label{Gap condizion}
    \inf_{k} \set{ |E_i(k) - E_c(k)| :\ i \in \In, c \notin \In}  >0,
\end{equation}
one considers the orthogonal projector
\begin{equation} \label{Eq:regularproj}
P_{\In}(k) = \sum_{n \in \In} \ket{u_n(k,\cdot)}\bra{u_n(k,\cdot)} \qquad  \in \BF(\Hf).
\end{equation}
It is known that, in view of condition \eqref{Gap condizion}, the map $k \mapsto P_{\In}(k)$ is a $\BF(\Hf)$-valued analytic and pseudoperiodic function (see, for example, \cite{Ne91} or \cite[Proposition 2.1]{Panati Pisante}). Thus, it defines a smooth vector bundle over the torus $\Base = \R^d/\Gamma^*$ whose fibre at $k \in \Base$ is $\Ran P_I(k)$.  The triviality of such vector bundle, called \emph{Bloch bundle} in \cite{Pa2008}, is equivalent to the existence of exponentially localized (composite) Wannier functions. By exploiting this geometric viewpoint, the existence of exponentially localized composite Wannier functions for a time-reversal-symmetric Hamiltonian has been proved, provided $d \leq 3$ \cite{Wannier_letter_BPCM, Pa2008}.

On the other hand, in metals and semimetals the relevant Bloch bands intersect each other. For example, in monolayer graphene the conduction and valence bands, here denoted by $E_+$ and $E_{-}$ respectively, form a conical intersection at two inequivalent points $K$ and $K'$ (Dirac points), \ie for $q = k - K$ one has
\begin{equation} \label{GrapheneBlochBands}
E_{\pm}(K + q) = \pm v_F |q| + o(|q|)  \qquad \mbox{as } |q| \to 0,
\end{equation}
where $v_F >0$ is called the Fermi velocity and $E_{\pm}(K) = 0$ by definition of the zero of the energy. An analogous expansion holds when $K$ is replaced by $K'$. The corresponding eigenprojectors
\[ P_s(k) =  \ket{u_s(k,\cdot)}\bra{u_s(k,\cdot)}, \qquad s \in \set{+,-}, \]
are not defined at $k=K$, nor does the limit $\lim_{k \to K} P_s(k)$ exist. This fact can be explicitly checked by using an effective tight-binding model
Hamiltonian \cite{Wallace,Goerbig,BenaMontambaux}.

%%%%%%%% SECTION 1.2 %%%%%%%%%%%%%%%%%%%%%%%%%%%%%%%%%%%%%%%%%%%%%%%%%%%%%%%%%%

\subsection{Tight-binding Hamiltonians in graphene}

We provide some details about the latter claim. If the Bloch functions $\psi_{s}(k)$ for the Hamiltonian $H_\Gamma$ were explicitly known, it would be natural to study the continuity of the eigenprojector by using the reduced 2-band Hamiltonian
\begin{equation}\label{Reduced Hamiltonian}
\Hi\sub{red}(k)_{r,s} = \inner{\psi_r(k)}{H_{\Gamma} \, \psi_s(k)}_{L^2(Y)} = \inner{ u_r(k)}{H(k) \, u_s(k)}_{\Hf}
\end{equation}
where $r,s \in \set{+,-}$, $Y$ is a fundamental cell for the lattice $\Gamma$ and $H(k)$ is defined in \eqref{H(k)}. Focusing on $|q| \ll 1$, one notices that, since $H(K)u_{\pm}(K)= E_{\pm}(K) u_{\pm}(K)=0$, a standard Hellman-Feynmann argument%
\footnote{The {Hellman-Feynmann-type argument} goes as follows. Near a Dirac point, \ie for $k = K+q$ and $|q| \ll 1$, one has
\begin{align*}
\Hi\sub{red}(k)_{r,s} &= \inner{ u_r(k)}{H(k) \, u_s(k)}_{\Hf} \\
&=  \inner{ u_r(K)}{H(K) \, u_s(K)}_{\Hf} +  \\
& \quad + q \cdot \left(\inner{ \nabla_k u_r(K)} {H(K) \, u_s(K)}_{\Hf} + \inner{H(K) \, u_r(K)} {\nabla_k u_s(K)}_{\Hf} \right) + \\
& \quad + q \cdot \inner{ u_r(K)}{\nabla_k H(K) \, u_s(K)}_{\Hf} + \Or(|q|^2).
\end{align*}
(Notice that the derivatives $\nabla_k H(k)$ exist in view of the above-mentioned analyticity of the family $k \mapsto H(k)$). Since $H(K)u_{\pm}(K)=0$, all terms but the last vanish. Clearly,
$$
\inner{u_r(K)}{H(K+q) \, u_s(K)} = q \cdot \inner{u_r(K)}{ \nabla_k H(K) \, u_s(K)} + \Or(|q|^2),
$$
yielding the claim.} %
yields
\begin{equation}\label{Fixed frame}
\Hi\sub{red}(K +q)_{r,s} = \inner{u_r(K)}{H(K+q) \, u_s(K)} + \Or(|q|^2).
\end{equation}
Thus $q \mapsto \Hi\sub{red}(K + q)$ encodes, for $|q| \ll 1$, the local behaviour of the Hamiltonian and its eigenprojectors with respect to \emph{fixed} Bloch functions, \ie Bloch functions evaluated at the Dirac point.

Approximated Bloch functions can be explicitly computed in the tight-binding approximation. Within this approximation, the reduced Hamiltonian \eqref{Reduced Hamiltonian} is approximated by the effective Hamiltonian
\begin{equation}\label{Effective Ham graphene gamma}
\Hi\sub{eff}(k) = \begin{pmatrix} 0 & \gamma_k^* \\ \gamma_k & 0 \end{pmatrix} \quad \text{where} \quad \gamma_k = 1 + \expo{\iu k \cdot a_2} + \expo{\iu k \cdot (a_2 - a_1)}
\end{equation}
with $\set{a_1, a_2}$ the standard Bravais basis for graphene, as in \cite{Wallace,Goerbig}.  Thus, for $|q| \ll 1$  and denoting by $\theta_q$ the polar angle in the plane $(q_1, q_2)$, \ie $|q| \expo{\iu \theta_q} = q_1 + \iu q_2$, one obtains
\begin{equation}\label{Effective Ham variant}
\Hi\sub{eff}(K +q) = v_{\rm F} \begin{pmatrix}
                          0              &  q_1 - \iu q_2 \\
                          q_1 + \iu q_2  &  0
                        \end{pmatrix} + \Or(|q|^2)
                   =  v_{\rm F} |q| \begin{pmatrix}
                          0              &  \expo{-\iu \theta_q} \\
                          \expo{\iu \theta_q} &  0
                        \end{pmatrix} + \Or(|q|^2).
\end{equation}
One easily checks (see Section \ref{sec:n=1}) that the eigenprojectors of the leading-order Hamiltonian, which is proportional to
\begin{equation}\label{Effective Ham graphene}
\Hi\sub{mono}(q) =  |q| \begin{pmatrix}
                          0              &  \expo{-\iu \theta_q} \\
                          \expo{\iu \theta_q} &  0
                        \end{pmatrix},
\end{equation}
are not continuous at $q=0$, implying that -- within the validity of the tight-binding approximation -- also the eigenprojectors of $H(k)$ are not continuous at $k =K$, as claimed. In the case of bilayer graphene, the same approach yields a leading-order effective Hamiltonian proportional to
\begin{equation}\label{Effective Ham bigraphene}
\Hi\sub{bi}(q) =  |q|^2 \begin{pmatrix}
                          0              &  \expo{-\iu 2 \theta_q} \\
                          \expo{\iu 2 \theta_q} &  0
                        \end{pmatrix}
\end{equation}
which also corresponds to a singular family of projectors. As pointed out by many authors \cite{Park-Marzari, McCann_Falko2006} the effective Hamiltonians \eqref{Effective Ham graphene} and \eqref{Effective Ham bigraphene} are related to different values of a \virg{topological index}. A rigorous definition of that index is the first task of our paper.

For the sake of completeness, we mention that according to \cite{MinMacDonald} the low-energy effective Hamiltonian for multilayer graphene (with $m$ layers and ABC stacking) is proportional to
\begin{equation}\label{Effective Ham multigraphene}
\Hi\sub{multi}(q) =   |q|^m \begin{pmatrix}
                          0              &  \expo{-\iu m \theta_q} \\
                          \expo{\iu m \theta_q} &  0
                        \end{pmatrix}, \qquad \qquad m \in \N^{\times}.
\end{equation}

%%%%%%%% SECTION 1.2 %%%%%%%%%%%%%%%%%%%%%%%%%%%%%%%%%%%%%%%%%%%%%%%%%%%%%%%%%%

\subsection{Singular families of projectors}

Abstracting from the specific case of graphene, we study the topology of the Bloch eigenspaces around an eigenvalue intersection. We consider any periodic Hamiltonian and a selected Bloch band of interest $E_s$ which intersects the other Bloch bands in finitely many points $K_1, \ldots, K_M$. Focusing on one of them, named $k_0$, the crucial information is the behaviour of the function $k \mapsto P_s(k)$ in a neighbourhood $R \subset \Base$ of the intersection point $k_0$. In view of the example of graphene, we set the following

\begin{df} \label{Def:singularproj}
Let $\Hi$ be a separable Hilbert space. A family of orthogonal projectors $\set{P(k)}_{k \in R \setminus \set{k_0}} \subset \BF(\Hi)$ such that $k \mapsto P(k)$ is $C^{\infty}$-smooth in $\dot R =  R \setminus \set{k_0} \subset \Base$ is called a \textbf{singular family} if it cannot be continuously extended to $k = k_0$, \ie if $\lim_{k \to k_0}P(k)$ does not exist. In such a case, the point $k_0$ is called \emph{singular point}. A family which is not singular is called \emph{regular family}.
\end{df}

In the case of insulators, the geometric structure corresponding to the regular family of projectors \eqref{Eq:regularproj} is a smooth vector bundle, whose topological invariants can be investigated with the usual tools of differential geometry (curvature, Chern classes, \ldots). On the other hand, in the case of metals and semimetals one deals with a singular family of projectors, thus the usual geometric approach is not valid anymore. Indeed, in the case of conical intersections the interesting information is \virg{hidden} in the singular point. In particular, if we assume for simplicity that the neighbourhood $R$ is a small ball around the Dirac point $k_0$, then when $d=2$ the set $R \setminus \set{k_0}$ can be continuously retracted to a circle $S^1$, and it is well known that every complex line bundle over $S^1$ is trivial, so it has no non-trivial topological invariants. In view of that, to define the topological invariants of a singular family of projectors we have to follow a different strategy, which is the content of Section \ref{sec:topology}. These invariants are also related to a distributional approach to the Berry curvature, as detailed in Appendix \ref{AppB}.

%%%%%%%% SECTION 2 %%%%%%%%%%%%%%%%%%%%%%%%%%%%%%%%%%%%%%%%%%%%%%%%%%%%%%%%%%

\section{\texorpdfstring{Topology of a singular $2$-dimensional family of projectors}{Topology of a singular 2-dimensional family of projectors}} \label{sec:topology}

Motivated by the example of graphene, in this paper we investigate the case of singular $2$-dimensional families of projectors ($d=2$, arbitrary rank).  Since the relevant parameter is the codimension, this case corresponds to the generic case in the Born-Oppenheimer theory of molecules \cite{Hagedorn1,Hagedorn2,Fermanian Lasser}. The analysis of higher-codimensional cases will be addressed in a future paper.

For the moment being, we focus on just one singularity of the family of projectors, say at the point $k_0$. We will conduct a \emph{local} analysis on the ``topological behaviour'' determined by such a singular point: hence we restrict our attention on a simply-connected region $R \subset \T^*_2$ containing $k_0$ such that for any other Bloch band $E_n$
\[ E_s(k) \ne E_n(k) \quad \text{for all } k \in R \setminus \{k_0\}. \]
In other words, Bloch bands are allowed to intersect only at $k_0$ in the region $R$.

%%%%%%%% SECTION 2.1 %%%%%%%%%%%%%%%%%%%%%%%%%%%%%%%%%%%%%%%%%%%%%%%%%%%%%%%%%%

\subsection{\texorpdfstring{A geometric $\Z$-invariant: eigenspace vorticity}{A topological Z-invariant: eigenspace vorticity}} \label{sec:vorticity}

To define a local geometric invariant characterising the behaviour of a $2$-dimensional family of projectors around the (possibly singular) point $k_0$, we start from the following

\begin{datum} \label{DatumProj}
Let $\Hi$ be a separable Hilbert space, and let $R \subset \T^*_2$ be a simply connected region containing $k_0$. We consider a family of projectors $\set{P_s(k)}_{k \in R \setminus \{k_0\}} \subset \mathcal{B}(\Hi)$ which is $C^\infty$-smooth in $R \setminus \{k_0\}$.
\end{datum}

We restrict our attention to the local behaviour of the family $\set{P_s(k)}$ around $k_0$. Suppose $r>0$ is so small that $U := \set{k \in \R^2 : |k - k_0| < r}$ is all contained in $R$. In order to define an integer-valued (local) geometric invariant $\nv$, baptised \emph{eigenspace vorticity},we provide the following computable recipe. First, introduce a smoothing parameter $\mu \in [-\mu_0, \mu_0]$, $\mu_0 > 0$, so that $\set{P_s(k)}_{R \setminus U}$ can be seen as the $\mu = 0$ case of a \emph{deformed} family of projectors $\set{P_s^\mu(k)}$, which for $\mu \ne 0$ is defined and regular on the whole region $R$. We also assume that the dependence on $\mu$ of such a deformed family of projectors is at least of class $C^2$.

The deformed family $\set{P_s^{\mu}(k)}$ allows us to construct a vector bundle $\calL_s$, which we call the \emph{smoothed Bloch bundle}, over the set
\[ B := \left( R \times [-\mu_0, \mu_0] \right) \setminus C, \]
where $C$ denotes the ``cylinder'' $C := U \times (-\mu_0, \mu_0)$. The total space of this vector bundle is
\begin{equation} \label{Ls} \calL_s := \set{((k,\mu),v) \in B \times \Hi : v \in \Ran P_s^{\mu}(k)}. \end{equation}
Explicitly, the fibre of $\calL_s$ over a point $(k,\mu) \in B$ is the range of the projector $P_s^{\mu}(k)$. We may look at $\calL_s$ as a collection of ``deformations'' of the bundle
\[ \calL_s^0 := \set{(k,v) \in (R \setminus U) \times \Hi : v \in \Ran P_s(k)} \]
over $R \setminus U$, which is defined solely in terms of the undeformed family $\set{P_s(k)}$.

We denote by $\omega_{s}$ the Berry curvature for the smoothed Bloch bundle $\calL_s$. Posing for notational convenience $k_3 = \mu$ and $\partial_j = {\partial}/{\partial k_j}$, one has
\begin{equation}\label{Berry curvature}
\omega_s = \iu \sum_{j,\ell = 1}^{3} (\omega_s)_{j,\ell}(k) \, \di k_j \wedge \di k_{\ell} \quad \mbox{ where } (\omega_s)_{j,\ell}(k) = \tr \left(P_s^{\mu}(k) \left[ \partial_j P_s^{\mu}(k), \partial_{\ell} P_s^{\mu}(k) \right] \right).
\end{equation}
Set $\mathcal{C} := \partial C$ for the ``cylindrical'' internal boundary of the base space $B$, and
\[ \dot{C} := C \setminus \left\{ \left( k = k_0, \mu = 0 \right) \right\} \]
for the ``pointed cylinder'', which is the complement of $B$ in
\[ \widehat{B} := \left( R \times [-\mu_0, \mu_0] \right) \setminus \left\{ \left( k = k_0, \mu = 0 \right) \right\}. \]

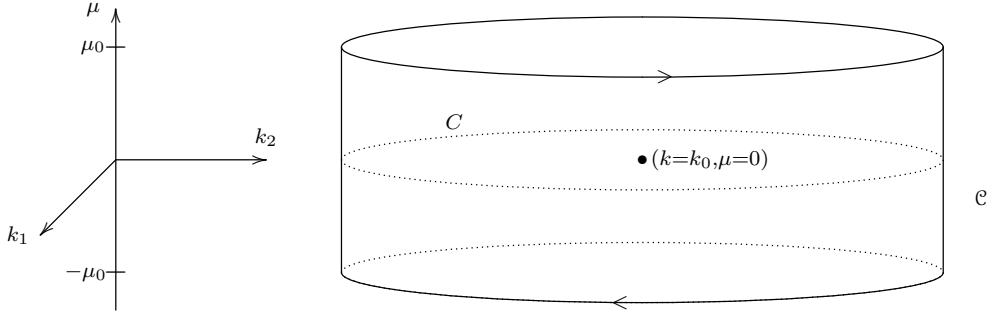
\begin{figure}[ht]
\centering
\begin{xy}
(-40,15); (0,15), {\ellipse(,.1){-}};
(-40,0); (0,0), {\ellipse(,.1){.}};
(-40,-15); (0,-15), {\ellipse(,.1){.}};
(-40,-15); (0,-15), {\ellipse(,.1)__{-}};
(-40,15); (-40,-15) **\dir{-};
(40,15); (40,-15) **\dir{-};
(0,0)*{\scriptstyle{\bullet}}; (9,0)*{\scriptstyle{(k=k_0,\mu=0)}};
(3,11)*{\scriptstyle{>}}; (-3,-19)*{\scriptstyle{<}};
(-70,-20); (-70,20) **\dir{-} ?>*\dir{>}; (-73,20)*{\scriptstyle{\mu}};
(-70,0); (-80,-10) **\dir{-} ?>*\dir{>}; (-83,-10)*{\scriptstyle{k_1}};
(-70,0); (-50,0) **\dir{-} ?>*\dir{>}; (-50,3)*{\scriptstyle{k_2}};
(-70,15)*{-}; (-73,15)*{\scriptstyle{\mu_0}};
(-70,-15)*{-}; (-74,-15)*{\scriptstyle{-\mu_0}};
(-25,5)*{\scriptstyle{C}}; (45,-5)*{\scriptstyle{\mathcal{C}}};
\end{xy}
\caption{The cylinder $C$ and its surface $\mathcal{C}$. The boundary of this surface is oriented according to the outward normal direction.}
\label{fig:cylinder}
\end{figure}

\begin{df}[\textbf{Eigenspace vorticity}] \label{Def:vorticity}
Let $\set{P_s(k)}_{k \in R \setminus \{k_0\}}$ be a family of projectors as in Datum \ref{DatumProj}, and $\set{P_s^{\mu}(k)}_{(k,\mu) \in B}$ be a smoothed family of projectors, as described above. Let $\calL_s$ be the vector bundle \eqref{Ls} over $B$, and denote by $\omega_{s}$ its Berry curvature, as in \eqref{Berry curvature}. The \textbf{eigenspace vorticity} of the smoothed family of projectors $\set{P_s^{\mu}(k)}_{k \in R \setminus \{k_0\}}$ (around the point $k_0 \in R$) is the integer
\begin{equation} \label{n=ch1(L)}
\nv = \nv(P_s) := - \frac{1}{2\pi} \int_{\mathcal{C}} \omega_{s} \in \Z.
\end{equation}
\end{df}

The number $\nv$ is indeed an integer, because it equals (up to a conventional sign) the \emph{first Chern number} of the vector bundle $\calL_s \to \mathcal{C}$. This integer is the topological invariant that characterises the behaviour of $\set{P^{\mu}_s(k)}$ around $k_0$; in particular, when $\dim \Ran P(k) = 1$ for $k \ne k_0$, it selects one of the canonical models that will be introduced in the next Subsection. 
Notice that there is an ambiguity in the \emph{sign} of the integer $\nv$ defined in \eqref{n=ch1(L)}, related to the orientation of the cylindrical surface of integration $\mathcal{C}$. Indeed, if one exchanges $\mu$ with $-\mu$ one obtains the opposite value for $\nv$. This ambiguity is resolved once the orientation of the $\mu$-axis is fixed.

\begin{rmk}[\textbf{Deformations in computational physics}] \label{Rem computational physics}
The above smoothing procedure corresponds to a common practice in computational solid-state physics%
\footnote{We are grateful to R. Bianco and R. Resta for pointing out this fact to us.}. %
Indeed, the family $\set{P_s(k)}$ usually appears as the collection of the eigenprojectors corresponding to the eigenvalue $E_s(k)$ of a given Hamiltonian operator $H(k)$ (\eg the Hamiltonian \eqref{Reduced Hamiltonian} or \eqref{Effective Ham graphene} for graphene). When dealing with an intersection of eigenvalues, to improve the numerical accuracy it is often convenient to consider a family of deformed Hamiltonians $H^{\mu}(k)$ and the corresponding eigenprojectors $P_s^{\mu}(k)$, in such a way that the eigenvalue intersection disappears when $\mu \neq 0$. For example, when dealing with monolayer graphene the deformed Hamiltonian is obtained by varying the electronegativity of the carbon atoms in the numerical code.
\end{rmk}

Before we proceed, we have to comment on the well-posedness of our definition of eigenspace vorticity. First of all, the definition relies on the existence of a deformation $\set{P_s^{\mu}(k)}$ of the original family $\set{P_s(k)}$ as in Datum \ref{DatumProj}, which for $\mu \ne 0$ is regular also at $k = k_0$. Such a deformation indeed exists in all cases of practical interest, namely when $P_s(k)$ arises as the eigenprojector,  relative to an eigenvalue $E_s(k)$, of a $k$-dependent Hamiltonian $H(k)$, such that at $k = k_0$ two of its eigenvalues coincide. Indeed, in view of the von Neumann-Wigner theorem \cite{Agrachev,von Neumann}, the eigenvalue intersection for general Hermitian matrices is highly non-generic, \ie it is a codimension-$3$ phenomenon. As the base $B$ of the smoothed Bloch bunde is $3$-dimensional, it may be assumed that the eigenvalue intersection -- \ie the singularity of the family of projectors -- occurs only at the point $(k = k_0, \mu = 0)$; more precisely, the generic
deformation $\set{P_s^{\mu}(k)}$ (which may be assumed to be the family of eigenprojectors of some deformed Hamiltonian $H^\mu(k)$) will satisfy this hypothesis.

In addition to this, we must investigate how the definition \eqref{n=ch1(L)} of $\nv$ depends on the specific choice of the deformation. Since $\nv$ is a topological quantity, it is stable under small perturbations of the deformed family of projectors $\set{P_s^{\mu}(k)}$. This means that model Hamiltonians which are ``close'', in some suitable sense, will produce the same eigenspace vorticity. More formally, we argue as follows. Let $\widetilde{P}_s^{\mu}(k)$ be another smoothing deformation of the family of projectors $P_s(k)$, so that in particular
\[ P_s^{\mu = 0}(k) = P_s(k) = \widetilde{P}_s^{\mu = 0}(k), \quad \text{for all } k \in R \setminus \set{k_0}. \]
Moreover, define the $2$-form $\widetilde{\omega}_s$ and its components $\left( \widetilde{\omega}_s \right)_{j,\ell}(k)$, $j,\ell \in \set{1, 2, 3}$, as in \eqref{Berry curvature}, with $P_s^{\mu}(k)$ replaced by $\widetilde{P}_s^{\mu}(k)$.

\begin{lemma}[\textbf{Irrelevance of the choice among close deformations}] \label{nvwelldefined}
Suppose that the maps $\widehat{B} \ni (k,\mu) \mapsto P_s^{\mu}(k) \in \mathcal{B}(\Hi)$ and $\widehat{B} \ni (k,\mu) \mapsto \widetilde{P}_s^{\mu}(k) \in \mathcal{B}(\Hi)$ are of class $C^2$, and that
\begin{equation} \label{P-P<1}
\left\| P^{\mu}(k) - \widetilde{P}^{\mu}(k) \right\|_{\mathcal{B}(\Hi)} < 1 \quad \text{for all } (k,\mu) \in \mathcal{C}.
\end{equation}
Then\begin{equation} \label{Omega_equality}
\int_{\mathcal{C}} \omega_s = \int_{\mathcal{C}} \widetilde{\omega}_s.
\end{equation}
\end{lemma}

\begin{proof}
By a result of Kato and Nagy \cite[Sec. I.6.8]{Kato}, the hypothesis \eqref{P-P<1} implies that there exists a family of unitary operators $W^{\mu}(k)$ such that
\begin{equation} \label{W_intertwines}
\widetilde{P}_s^{\mu}(k) = W^{\mu}(k) P_s^{\mu}(k) W^{\mu}(k)^{-1}. 
\end{equation}
The explicit Kato-Nagy's formula
\begin{equation} \label{W definition}
W^{\mu}(k) := \left( \1 - (P_s^{\mu}(k) - \widetilde{P}_s^{\mu}(k))^2 \right)^{-1/2} \left( \widetilde{P}_s^{\mu}(k) P_s^{\mu}(k) + (\1 - \widetilde{P}_s^{\mu}(k))(\1 - P_s^{\mu}(k))\right)
\end{equation}
shows that the map $(k, \mu) \mapsto W^{\mu}(k)$ has the same regularity as  $(k, \mu) \mapsto P_s^{\mu}(k) - \widetilde{P}_s^{\mu}(k)$%
\footnote{The presence of the inverse square root does not spoil the regularity of $W^{\mu}(k)$. Indeed, setting $Q = Q^{\mu}(k) = (P_s^{\mu}(k) - \widetilde{P}_s^{\mu}(k))^2$, one can expand
\begin{equation} \label{powerseries} (\1 - Q)^{-1/2} = \sum_{n=0}^{\infty} \binom{-1/2}{n} (-Q)^n. \end{equation}
The above power series is absolutely convergent if $\|Q\| < 1$ (which follows from \eqref{P-P<1}), and in the same range it is term-by-term differentiable.}. %%%End footnote

Formula \eqref{W_intertwines} implies that the vector bundles $\calL_s$ and $\widetilde{\calL}_s$, corresponding to the deformed families of projectors {$\set{{P}_s^{\mu}(k)}$ and $\{\widetilde{P}_s^{\mu}(k)\}$} respectively, are isomorphic, the isomorphism being implemented fibre-wise by the unitary $W^{\mu}(k)$. Thus they have the same Chern number, \ie equation \eqref{Omega_equality} holds true.

Alternatively,  for the sake of clarity, we provide an explicit proof of  {\eqref{Omega_equality}} by showing that the difference $\widetilde{\omega}_s - \omega_s$ is an exact form $\di \beta$ on $\mathcal{C}$; then by applying Stokes' theorem one gets
\[ \int_{\mathcal{C}} (\widetilde{\omega}_s - \omega_s) = \int_{\mathcal{C}} \di \beta = 0 \]
because $\mathcal{C}$ has no boundary. This will conclude the proof of the Lemma.

For the sake of readability, in the following we will use the abbreviations
\begin{gather*}
P := P_s^{\mu}(k), \quad \widetilde{P} := \widetilde{P}_s^{\mu}(k), \quad W := W^{\mu}(k), \\
\omega_{j,\ell} := \left( \omega_s \right)_{j,\ell}(k), \quad \widetilde{\omega}_{j,\ell} := \left( \widetilde{\omega}_s \right)_{j,\ell}(k).
\end{gather*}

A lengthy but straightforward computation, that uses only the cyclicity of the trace, the relations $P^2 = P$ and $W W^{-1} = \1 = W^{-1} W$ and their immediate consequences
\[ P (\partial_j P) = \partial_j P - (\partial_j P) P  \qquad \mbox{ and } \qquad  W^{-1} (\partial_j W) = - (\partial_j W^{-1}) W, \]
yields to
\begin{equation} \label{Contazzo1}
\begin{aligned}
\widetilde{\omega}_{j,\ell}  - \omega_{j,\ell} = & \tr \left\{ P (\partial_j W^{-1}) (\partial_{\ell} W) - P (\partial_{\ell} W^{-1}) (\partial_j W) \right\} \\
& + \tr \left\{ (\partial_j P) W^{-1} (\partial_{\ell} W) - (\partial_{\ell} P) W^{-1} (\partial_j W) \right\}.
\end{aligned}
\end{equation}
Summing term by term the two lines in \eqref{Contazzo1}, one gets
\begin{align*}
\widetilde{\omega}_{j,\ell}  - \omega_{j,\ell}
& = \tr \left\{ \left( P (\partial_j W^{-1}) + (\partial_j P) W^{-1} \right) (\partial_{\ell} W) - \left( P (\partial_{\ell} W^{-1}) + (\partial_{\ell} P) W^{-1} \right) (\partial_j W) \right\} = \\
& = \tr \left\{ \partial_j (P W^{-1}) \partial_{\ell} W - \partial_{\ell} (P W^{-1}) \partial_j W \right\} = \\
& = \tr \left\{ \partial_j \left( P W^{-1} (\partial_{\ell} W) \right) - P W^{-1} \partial_j \partial_{\ell} W  - \partial_{\ell} \left( P W^{-1} (\partial_j W) \right) + P W^{-1} \partial_{\ell} \partial_j W \right\} = \\
& = \partial_j \left( \tr \left\{ P W^{-1} (\partial_{\ell} W) \right\} \right) - \partial_{\ell} \left( \tr \left\{ P W^{-1} (\partial_j W) \right\} \right),
\end{align*}
where we used the fact that $(k, \mu) \mapsto W^{\mu}(k)$ is at least of class $C^2$, so that the mixed second derivatives cancel. In summary, the explicit computation shows that
\[ \widetilde{\omega}_{j,\ell} - \omega_{j,\ell} = \partial_j \tr \left( P_s^{\mu}(k) W^{\mu}(k)^{-1} \partial_{\ell} W^{\mu}(k) \right) - \partial_{\ell} \tr \left( P_s^{\mu}(k) W^{\mu}(k)^{-1} \partial_j W^{\mu}(k) \right) \]
which can be written in a more intrinsic form as
\[ \widetilde{\omega}_s - \omega_s = \di \beta, \quad \text{with} \quad \beta = 2 \, \iu \, \tr \left( P_s^{\mu}(k) W^{\mu}(k)^{-1} \di W^{\mu}(k) \right). \]
This concludes the proof.
\end{proof}

\begin{rmk}[\textbf{Numerical evaluation of the eigenspace vorticity}] \label{Numerical evaluation}
The numerical evaluation of $n\sub{v}$ can be performed by replacing the cylindrical surface $\mathcal{C}$ with any surface homotopically equivalent to it in $\widehat{B}$, \eg with any polyhedron enclosing the point $\left(k_0,0\right)$. Numerically, the eigenspace vorticity is evaluated by summing up contributions from all faces of the polyhedron. The integral of the curvature $\omega_s$ over each face is computed by a discretization scheme which approximates the integral of the Berry connection over the perimeter of the face (see \eg \cite{Park-Marzari} and references therein). The latter approach has been implemented by R. Bianco in the case of the Haldane model \cite{Haldane88}, and provided results in agreement with the analytical computation already in an $8$-point discretization, by using a cube.
\end{rmk}

Even with the above result, the value of $\nv$ may still {\it a priori} depend on the choice of the specific deformation $\set{P_s^{\mu}(k)}$, and not only on the original family of projectors $\set{P_s(k) = P_s^{\mu=0}(k)}$. However, as we will explain in Section \ref{sec:pseudospin}, when $\Hi = \C^2$ there exists a class of ``distinguished deformations'', called \emph{hemispherical}, which provide a natural choice of deformation to compute the eigenspace vorticity. For such hemispherical deformations, the eigenspace vorticity indeed depends only on the undeformed family of projectors, \ie on the Datum \ref{DatumProj}. In all the relevant examples, as the tight-biding model of graphene or the Haldane model, such hemispherical deformations appear naturally.

For the sake of clarity, in the next Subsection we introduce a family of \emph{canonical models}, one for each value of $\nv \in \Z$,  having the property of being hemispherical. For the case of a general Hilbert space $\Hi$, the identification of a ``distinguished'' class of deformations is a challenging open problem.

%%%%%%%% SECTION 2.2 %%%%%%%%%%%%%%%%%%%%%%%%%%%%%%%%%%%%%%%%%%%%%%%%%%%%%%%%%%

\subsection{The canonical models for an intersection of eigenvalues} \label{sec:models}

In this Subsection, we introduce effective Hamiltonians whose eigenspaces model the topology of Bloch eigenspaces, locally around a point $k_0$ where Bloch bands intersect. These will be also employed as an example on how to perform the $\mu$-deformation for a family of eigenprojectors.

Hereafter, we focus on the case of a system of two \emph{non-degenerate} Bloch bands, \ie $P_s(k)$ is a projector on $\Hi = \C^2$ and $\dim \Ran P_s(k)=1$ for $k \ne k_0$.

%%%%%%%% SECTION 2.2.1 %%%%%%%%%%%%%%%%%%%%%%%%%%%%%%%%%%%%%%%%%%%%%%%%%%%%%%%%%%

\subsubsection{\texorpdfstring{\textbf{The $1$-canonical model}}{The 1-canonical model}} \label{sec:n=1}

Assume from now on that $k \in U = \set{k \in \R^2 : |k-k_0|< r}$, where $r > 0$ is sufficiently small; as before, set $q = k - k_0$. We will mainly work in polar coordinates in momentum space, and hence denote by $(|q|, \ph)$ the coordinates of $q$, namely $q_1 + \iu q_2 = |q| \, \expo{\iu \ph}$.

Following \cite{Hagedorn1}, we consider the effective Hamiltonian\footnote{The Hamiltonian \eqref{Heff} is unitarily equivalent to the effective Hamiltonian for monolayer graphene \eqref{Effective Ham graphene}, by conjugation with a \emph{$k$-independent} unitary matrix. Thus, both \eqref{Effective Ham graphene} and \eqref{Heff} are reasonable choices for a canonical local model describing conical intersections of eigenvalues. We prefer the choice \eqref{Heff}, since this Hamiltonian has real entries.}%
\begin{equation} \label{Heff}
H\sub{eff}(q) := \begin{pmatrix} q_1 & q_2 \\ q_2 & -q_1 \end{pmatrix} = |q| \begin{pmatrix} \cos \ph & \sin \ph \\ \sin \ph & - \cos \ph \end{pmatrix}.
\end{equation}
The eigenvalues of this matrix are given by
\[ E_{\pm}(q) = \pm  |q| \]
and thus $H\sub{eff}(q)$ is a good candidate for modelling conical intersections locally (compare \eqref{GrapheneBlochBands}). The eigenfunctions corresponding to $E_{+}(q)$ and $E_{-}(q)$ are respectively
\begin{equation} \label{canonic_conical}
\phi_{1,+}(q) = \eu^{\iu \ph/2} \begin{pmatrix} \cos(\ph/2) \\ \sin(\ph/2) \end{pmatrix}, \quad \phi_{1,-}(q) = \eu^{\iu \ph/2} \begin{pmatrix} - \sin(\ph/2) \\ \cos(\ph/2) \end{pmatrix}.
\end{equation}
The phases are chosen so that these functions are single-valued when we identify $\ph = 0$ and $\ph = 2 \pi$. We will call $\phi_{1,\pm}(q)$ the \emph{canonical eigenvectors for the conical intersection} at the singular point $k_0$. These satisfy
\begin{equation} \label{nablaphi}
\partial_{|q|} \phi_{1,\pm} = 0, \quad \partial_{\ph} \phi_{1,\pm} = \frac{1}{2} \left( \pm \phi_{1,\mp} + \iu \phi_{1,\pm} \right).
\end{equation}
The corresponding eigenprojectors are easily computed to be
\begin{equation} \label{ConicalProj}
P_{1,\pm}(q) = \pm \frac{1}{2|q|} \begin{pmatrix} q_1 \pm |q| & q_2 \\ q_2 & - q_1 \pm |q| \end{pmatrix} = \pm \frac{1}{2} \begin{pmatrix} \cos \ph \pm 1 & \sin \ph \\ \sin \ph & - \cos \ph \pm 1 \end{pmatrix}.
\end{equation}

The above expressions show that the families $\{P_{1,\pm}(q)\}_{q \in \dot{U}}$ are singular at $q=0$, in the sense of Definition \ref{Def:singularproj}.  As explained in Section \ref{sec:vorticity}, in order to compute their vorticities we have to introduce a smoothing parameter $\mu \in [-\mu_0, \mu_0]$ to remove the singularity: this is achieved by considering the so-called \emph{avoided crossings} \cite{Hagedorn2}. Explicitly, we deform the Hamiltonian \eqref{Heff} to get
\begin{equation} \label{Heffmu}
H\sub{eff}^\mu(q) := \begin{pmatrix} q_1 & q_2 + \iu \mu \\ q_2 - \iu \mu & -q_1 \end{pmatrix} = |q| \begin{pmatrix} \cos \ph & \sin \ph + \iu \eta \\ \sin \ph - \iu \eta & - \cos \ph \end{pmatrix},
\end{equation}
where $\mu \in [- \mu_0, \mu_0]$ is a small real parameter and
\[ \eta = \eta^\mu(q) := \frac{\mu}{|q|}. \]
In this case, the eigenvalues are
\[ E_{\pm}^\mu(q) := \pm \sqrt{|q|^2+\mu^2} = \pm |q| \sqrt{1+\eta^2} \]
and thus, if $\mu \ne 0$, the energy bands do not intersect.

As the matrix $H\sub{eff}^\mu(q)$ is not real, we look for complex eigenfunctions of the form $v + \iu u$. These can be found by solving the system
\begin{equation} \label{system}
\begin{cases}
\cos \ph \, v_1 + \sin \ph \, v_2 = \pm \sqrt{1+\eta^2} \, v_1 + \eta \, u_2, \\
- \sin \ph \, v_1 + \cos \ph \, v_2 = \mp \sqrt{1+\eta^2} \, v_2 + \eta \, u_1, \\
\cos \ph \, u_1 + \sin \ph \, u_2 = \pm \sqrt{1+\eta^2} \, u_1 - \eta \, v_2, \\
- \sin \ph \, u_1 + \cos \ph \, u_2 = \mp \sqrt{1+\eta^2} \, u_2 - \eta \, v_1,
\end{cases}
\end{equation}
with respect to the unknowns $(v_1, v_2, u_1, u_2)$, and then imposing that they coincide, up to the phase $\eu^{\iu \ph/2}$, with \eqref{canonic_conical} when $\mu = 0$ and $q \ne 0$ (a condition on $(v_1, v_2)$). After choosing the appropriate phase, one gets the \emph{canonical eigenvectors for the avoided crossing}
\begin{equation} \label{canonic_avoided}
\phi_{1,\pm}^\mu(q) := \frac{1}{\sqrt{1+\alpha^2}} \left[ \phi_{1,\pm}(q) + \iu \alpha \phi_{1,\mp}(q) \right]
\end{equation}
where
\begin{equation} \label{alpha}
\alpha = \alpha^\mu(q) := \frac{1-\sqrt{1+\eta^2}}{\eta} = \frac{|q|-\sqrt{|q|^2+\mu^2}}{\mu}.
\end{equation}
One can easily check that the eigenprojectors associated to these eigenvectors are
\begin{align*}
P_{1,\pm}^\mu(q) & = \pm \frac{1}{2 \sqrt{|q|^2 + \mu^2}} \begin{pmatrix} q_1 \pm \sqrt{|q|^2 + \mu^2} & q_2 + \iu \mu \\ q_2 - \iu \mu & - q_1 \pm \sqrt{|q|^2 + \mu^2} \end{pmatrix} = \\
& = \pm \frac{1}{2 \sqrt{1+\eta^2}} \begin{pmatrix} \cos \ph \pm \sqrt{1+\eta^2} & \sin \ph + \iu \eta \\ \sin \ph - \iu \eta & - \cos \ph \pm \sqrt{1+\eta^2}
\end{pmatrix}.
\end{align*}

It is convenient to describe the behavior of the eigenspaces in geometric terms. For a fixed choice of the index $s \in \{+, -\}$, we introduce a line bundle $\PB_{1,s}$, called the \emph{\Composite bundle}, on the pointed cylinder $\dot{C}$,
whose fibre at the point $(q,\mu) \in \dot{C}$ is just the range of the projector $P_{1,s}^\mu(q)$.
As the \Composite bundle $\PB_{1,s}$ is (up to retraction of the basis)  a line bundle over the $2$-dimensional manifold $\mathcal{C}$, it is completely
characterised by its first Chern number. The latter can be computed as the integral over $\mathcal{C}$ of the Berry curvature, which can be interpreted as an ``inner'' vorticity of the family $\set{P_{1,s}(q)}$ (compare Definition
\ref{Def:vorticity}).

Explicitly, by using polar coordinates, the Berry curvature of the \Composite bundle $\PB_{1,\pm}$ reads
\begin{equation}  \label{omegapmmu}
\begin{aligned}
\omega_{1,\pm} = & -2 \left( \Im \scal{\partial_{|q|} \phi_{1,\pm}^\mu(q)}{\partial_{\ph} \phi_{1,\pm}^\mu(q)}_{\Hf} \di |q| \wedge \di \ph \, + \right. \\
& + \Im \scal{\partial_{|q|} \phi_{1,\pm}^\mu(q)}{\partial_{\mu} \phi_{1,\pm}^\mu(q)}_{\Hf} \di |q| \wedge \di \mu \, + \\
& + \left. \Im \scal{\partial_{\ph} \phi_{1,\pm}^\mu(q)}{\partial_{\mu} \phi_{1,\pm}^\mu(q)}_{\Hf} \di \ph \wedge \di \mu \right).
\end{aligned}
\end{equation}
By using \eqref{nablaphi}, one computes the derivatives appearing in the above expression: this yields to
\begin{equation} \label{derivatives}
\begin{aligned}
\partial_{|q|} \phi_{1,\pm}^\mu & = \frac{\iu}{1+\alpha^2} \, \phi_{1,\mp}^\mu \partial_{|q|} \alpha, \\
\partial_{\mu} \phi_{1,\pm}^\mu & = \frac{\iu}{1+\alpha^2} \, \phi_{1,\mp}^\mu \partial_{\mu} \alpha, \\
\partial_{\ph} \phi_{1,\pm}^\mu & = \frac{1}{2} \frac{1 \mp \alpha}{\sqrt{1 + \alpha^2}} \, \left( \pm \phi_{1,\mp} + \iu \phi_{1,\pm} \right).
\end{aligned}
\end{equation}
As
\[ \scal{\phi_{1,\mp}^\mu(q)}{\pm \phi_{1,\mp}(q) + \iu \phi_{1,\pm}(q)}_{\Hf} = \pm \frac{1 \pm \alpha}{\sqrt{1+\alpha^2}} \]
one obtains
\begin{equation} \label{omega1}
\omega_{1,\pm} = \pm \frac{1}{2} \left[ \partial_{|q|} \left( \frac{\mu}{\sqrt{|q|^2+\mu^2}} \right) \di |q| \wedge \di \ph - \partial_{\mu} \left( \frac{\mu}{\sqrt{|q|^2+\mu^2}} \right) \di \ph \wedge \di \mu \right].
\end{equation}
Integrating the curvature of the Berry connection over the surface $\mathcal{C}$, one obtains the Chern number
\[ \ch_1(\PB_{1,\pm}) = \frac{1}{2 \pi} \int_{\mathcal{C}} \omega_{1,\pm} = \mp 1 \]
or equivalently
\[ \nv(P_{1,\pm}) = \pm 1. \]

%%%%%%%% SECTION 2.2.2 %%%%%%%%%%%%%%%%%%%%%%%%%%%%%%%%%%%%%%%%%%%%%%%%%%%%%%%%%%

\subsubsection{\texorpdfstring{\textbf{The $n$-canonical model}}{The n-canonical model}} \label{sec:n=n}

We now exhibit model Hamiltonians $H_n(q)$, having Bloch bands $E_{\pm}(q) = \pm e(q)$ (with $e(0) = 0$ and $e(q) > 0$ for $q \ne 0$, in order to have eigenvalue intersections only at $q=0$), such that the corresponding \Composite bundles have a Chern number equal to an arbitrary $n \in \Z$, and that in particular $H\sub{eff}(q) = H_{n=1}(q)$ when we choose $e(q) = |q|$ (\ie when a conical intersection of  bands is present). Notice that, if $P_{n, \pm}(q)$ are the eigenprojectors of the Hamiltonian $H_n(q)$, then
\begin{equation} \label{rev_eng}
H_n(q) = E_+(q) P_{n, +}(q) +  E_-(q) P_{n, -}(q).
\end{equation}
Thus, it suffices to provide an {\it ansatz} for the eigenfunctions $\phi_{n, \pm}(q)$ of $H_n(q)$. Set
\begin{equation} \label{n-canonical}
\phi_{n, +}(q) = \eu^{\iu n \ph/2} \begin{pmatrix} \cos(n \ph/2) \\ \sin(n \ph/2) \end{pmatrix}, \quad \phi_{n, -}(q) = \eu^{\iu n \ph/2} \begin{pmatrix} - \sin(n \ph/2) \\ \cos(n \ph/2)
\end{pmatrix}.
\end{equation}
Notice that, for even $n$, the functions $\cos(n \ph/2)$ and $\sin(n \ph/2)$ are already single-valued under the identification of $\ph = 0$ with $\ph = 2 \pi$, but so is the phase $\eu^{\iu n \ph/2}$, so there is no harm in inserting it. These functions will be called the \emph{$n$-canonical eigenvectors}. As for their eigenprojectors, one easily computes
\begin{equation} \label{nProjectors}
P_{n, \pm}(q) = \pm \frac{1}{2} \begin{pmatrix} \cos n \ph \pm 1 & \sin n\ph \\ \sin n \ph & - \cos n \ph \pm 1 \end{pmatrix}. \end{equation}
By \eqref{rev_eng} above we get that the $n$-th Hamiltonian is
\begin{equation} \label{Hn}
H_n(q) = e(q) \begin{pmatrix} \cos n \ph & \sin n \ph \\ \sin n \ph & - \cos n \ph \end{pmatrix}
\end{equation}
provided that we show that the \Composite bundles $\PB_{n, \pm}$ have first Chern number equal to $\mp n$.

In order to evaluate $\ch_1(\PB_{n, \pm})$, we perturb the Hamiltonian $H_n(q)$ in a way which is completely analogous to what we did for $H\sub{eff}(q)$, and define
\begin{equation} \label{Hnmu}
H_n^\mu(q) := e(q) \begin{pmatrix} \cos n\ph & \sin n\ph + \iu \eta \\ \sin n\ph - \iu \eta & - \cos n\ph \end{pmatrix}, \quad \eta = \frac{\mu}{e(q)}.
\end{equation}
The eigenvalues of $H_n^{\mu}(q)$ are $E_{\pm}^\mu(q) = \pm e^{\mu}(q)$, where%
\footnote{One could also choose to put a different eigenvalue $\widetilde{e}(q, \mu)$ in front of the matrix in \eqref{Hnmu}, because the topology of the \Composite bundle depends only on the family of projectors $\set{P_{n, \pm}^{\mu}(q)}$ (which we will introduce in a moment). We will adhere to our definition of $e^{\mu}(q)$ in order to recover the model \eqref{Heffmu} when we set $n = 1$ and $e(q) = |q|$.} %
$e^{\mu}(q) := e(q) \sqrt{1+\eta^2}$.

Its eigenfunctions can be found by solving a system similar to \eqref{system}, obtained just by replacing $\ph$ with $n \ph$. After straightforward calculations, one eventually finds
\[ \phi_{n, \pm}^\mu(q) := \frac{1}{\sqrt{1+\alpha^2}} \left[ \phi_{n, \pm}(q) + \iu \alpha \phi_{n, \mp}(q) \right] \]
with the \emph{same} $\alpha$ as in \eqref{alpha}. One easily checks that the associated eigenprojectors are
\begin{equation} \label{nmuProjectors}
P_{n, \pm}^{\mu}(q) = \pm \frac{1}{2 \sqrt{1+\eta^2}} \begin{pmatrix} \cos n \ph \pm \sqrt{1+\eta^2} & \sin n \ph + \iu \eta \\ \sin n \ph - \iu \eta & - \cos n \ph \pm \sqrt{1+\eta^2}
\end{pmatrix}.
\end{equation}

Now notice that
\[ \partial_{|q|} \phi_{n, \pm} = 0, \quad \partial_{\ph} \phi_{n, \pm} = \frac{n}{2} \left( \pm \phi_{n, \mp} + \iu \phi_{n, \pm} \right). \]
Hence, in order to compute the Berry curvature of the \Composite bundle $\PB_{n, \pm}$, one only has to modify the expression for all the derivatives (and related scalar products with other derivatives) computed above substituting $\phi_{1,\pm}^\mu(q)$ with $\phi_{n, \pm}^\mu(q)$, and multiplying by $n$ the ones involving derivatives with respect to $\ph$. Notice that the dependence of $\eta$, and consequently of $P_{n, \pm}^{\mu}(q)$, on $e(q)$ does not affect this computation, provided the hypothesis $e(0)=0$ holds. Explicitly, this procedure yields to
\begin{equation} \label{omegan}
\omega_{n,\pm} = \pm \frac{n}{2} \left[ \partial_{|q|} \left( \frac{\mu}{\sqrt{e(q)^2+\mu^2}} \right) \di |q| \wedge \di \ph - \partial_{\mu} \left( \frac{\mu}{\sqrt{e(q)^2+\mu^2}} \right) \di \ph \wedge \di \mu \right],
\end{equation}
and correspondingly the ``inner'' vorticity of $\set{P_{n,\pm}(q)}$ equals
\[ \nv(P_{n, \pm}) = - \ch_1(\PB_{n, \pm}) = - \frac{1}{2 \pi} \int_{\mathcal{C}} \omega_{n,\pm} = \pm n, \]
as we wanted.

\begin{rmk}[{\bf The case $n = 0$}] \label{rmk:n=0}
When $n = 0$, the functions $\phi_{0, \pm}(q)$ are constant, and hence are defined on the whole disc  $U$. Correspondingly, the \Composite bundles $\PB_{0, \pm}$ are both isomorphic to the trivial line bundle $\left( U \times [-\mu_0, \mu_0] \right) \times \C$. Of course, the bundles $\PB_{0, \pm}$ do not correspond to families of \emph{singular} projectors; however, this notation will help us state our results in a neater way in the following.
\end{rmk}

%%%%%%%% SECTION 2.3 %%%%%%%%%%%%%%%%%%%%%%%%%%%%%%%%%%%%%%%%%%%%%%%%%%%%%%%%%%

\goodbreak

\subsection{Comparison with the pseudospin winding number} \label{sec:pseudospin}

In this Section we compare the \emph{eigenspace vorticity} with the \emph{pseudospin winding number} (PWN) which appears in the literature about graphene \cite{Park-Marzari,McCann_Falko2006}. While the former is defined in a wider context, it happens that these two indices agree whenever the latter is well-defined, including the relevant cases of monolayer and multilayer graphene. We also show that the pseudospin winding number is neither a Berry phase, as already noticed in \cite{Park-Marzari}, nor a topological invariant.

%%%%%%%% SECTION 2.3.1 %%%%%%%%%%%%%%%%%%%%%%%%%%%%%%%%%%%%%%%%%%%%%%%%%%%%%%%%%%

\subsubsection{\textbf{Definition of the pseudospin winding number.}}
We firstly rephrase the usual definition in a more convenient language. The starting point is the following
\begin{datum} \label{Datum}
For $\Hi = \C^2$, let $\set{P(k)} \subset \Bi(\Hi)$ be a family of orthogonal projectors defined on the circle $S^1 = \partial U$, where $U =\set{k \in \R^2 : |k -k_0| < r}$, for a suitable $r>0$. We will also assume that the range of each projector $P(k)$ is $1$-dimensional, because this is clearly the only interesting case.
\end{datum}

Let $\Psi: S^1 \rightarrow \C^2$ be a continuous map such that $\Psi(k) \in \Ran P(k)$ and $\Psi(k) \neq 0$ for every $k \in S^1$. Such a map does exist: Indeed, let $\mathcal{P}$ be the line bundle over $S^1$ corresponding to $\set{P(k)}$, whose total space is
\[ \mathcal{P} = \set{(k,v) \in S^1 \times \C^2 : \ v \in \Ran P(k)}. \]
The line bundle $\mathcal{P}$ is trivial, since every \emph{complex} line bundle over $S^1$ is so. Therefore there exists a global non-zero continuous section of $\mathcal{P}$, here denoted by $\Psi$.  Without loss of generality, we assume $\norm{\Psi(k)} \equiv 1$.

\noindent With respect to a fixed orthonormal basis $\set{e_1, e_2} \subset \C^2$ one writes
\begin{equation}\label{Components}
\Psi(k) = \psi_1(k) e_1 + \psi_2(k) e_2.
\end{equation}
Obviously, $\psi_1(k)$ and $\psi_2(k)$ can not be simultaneously zero. One makes moreover the following (sometimes implicit) assumption.
\begin{assumption}\label{Assumption pseudospin}
Assume that for every $k \in S^1 = \set{k \in \R^2 : |k -k_0| = r}$ the numbers $\psi_1(k)$ and $\psi_2(k)$ are \emph{both} non-zero.
\end{assumption}
\noindent Under this assumption, which implies $\psi_j(k) = |\psi_j(k)| \, \expo{\iu \theta_j(k)}$ with $\theta_j$ continuous, the map
\[ g: S^1 \rightarrow U(1), \qquad g(k) = \mathrm{phase} \left( \frac{\psi_2(k)}{\psi_1(k)} \right) = \expo{\iu (\theta_2(k)-\theta_1(k))} \]
is well-defined and continuous. Then the \emph{pseudospin winding number} $n\sub{w} = n\sub{w}(P)$ is defined as the degree of the continuous map $g$. In terms of polar coordinates for $q = k-k_0$, namely $q = (|q|, \theta_q)$, one has
\begin{equation}\label{Pseudospin def}
n\sub{w} = n\sub{w}(P) := \deg{g} =   \frac{1}{2 \pi \iu} \oint_{S^1} \di k \, \overline{g(k)} \,\, \frac{\de g}{\de \theta_q} (k) \, \in \Z.
\end{equation}

For example, when we take the eigenprojectors $P_{\pm}(k)$ of the effective tight-binding Hamiltonians \eqref{Effective Ham graphene} and \eqref{Effective Ham multigraphene} as our datum, the Assumption \ref{Assumption pseudospin} is satisfied, with respect to the canonical basis of $\C^2$, by the global section (for the $m$-multilayer graphene Hamiltonian \eqref{Effective Ham multigraphene})
\begin{equation}\label{Eigenfunction multilayer}
\Psi_{m, s}(q) =  \frac{1}{\sqrt{2}}  \begin{pmatrix} 1  \\ s \, \expo{\iu m \theta_q} \end{pmatrix}, \qquad  m \in \N^{\times},
\end{equation}
where $s \in \set{+,-}$ refers to the choice of the upper (resp. lower) eigenvalue. It is then straightforward to check that
\begin{equation}\label{Pseudospin graphene}
n\sub{w} = \begin{cases}
               1 & \text{for monolayer graphene,} \\
               2 & \text{for bilayer graphene,} \\
               m & \text{for $m$-multilayer graphene.}
           \end{cases}
\end{equation}

\begin{rmk}[Comparison with the Berry phase] \label{Rem Comparison with Berry phase}
We emphasise that, despite the formal similarity of the definitions, in general the pseudospin winding number is \emph{not} a Berry phase, as clarified by Park and Marzari \cite{Park-Marzari}. Indeed, $g$ is not a wave function, but the ratio of the components of a single ($\C^2$-valued) wave function with respect to a chosen orthonormal basis.

On the other hand, if in some particular model it happens that
\begin{equation}\label{Equatorial range}
|\psi_2(k)/\psi_1(k)| = |g(k)| \equiv 1 \mbox{ for all } k \in S^1
\end{equation}
(as indeed happens in the case \eqref{Eigenfunction multilayer}), then the holonomy of the Berry connection $\mathcal{A}$ (\ie the Berry phase) along the circle $S^1$ is
\[ \mathrm{Hol}\,{\mathcal{A}} = \mathrm{exp} \left\{ - \oint_{S^1} \inner{\Psi(k)}{\frac{\partial}{\partial \theta_q} \Psi(k)}_{\C^2} \right\} = \expo{-\iu \pi n\sub{w}}, \]
as can be checked by direct computation, assuming without loss of generality that the first component of $\Psi(k)$ is real, \ie $\psi_1(k) = |\psi_1(k)|$. In such a case, $n\sub{w}$ contains a more detailed information than the Berry phase, which is defined only modulo $2 \pi \Z$.
\end{rmk}

\begin{rmk}[Hartree-Fock corrections]
It has been recently shown \cite{Lewin} that, when many-electron Coulomb interactions in monolayer graphene are taken into account, the leading order correction in the framework of the Hartree-Fock theory amounts to replace the effective Hamiltonian \eqref{Effective Ham graphene} with
\[ \Hi\sub{mono}^{\mathrm{(HF)}}(q) = v\sub{eff}(q) \, |q| \begin{pmatrix}
                          0              &  \expo{-\iu \theta_q} \\
                          \expo{\iu \theta_q} &  0
                        \end{pmatrix}, \]
where $v\sub{eff}$ is an explicit function \cite[Lemma 3.2]{Lewin}. The corresponding eigenprojectors coincide with those of the Hamiltonian \eqref{Effective Ham graphene}. Therefore, the PWN equals the one computed in the tight-binding model.
\end{rmk}

%%%%%%%% SECTION 2.3.2 %%%%%%%%%%%%%%%%%%%%%%%%%%%%%%%%%%%%%%%%%%%%%%%%%%%%%%%%%%

\subsubsection{\textbf{Geometrical reinterpretation.}}
We find convenient to reinterpret the definition of $n\sub{w}$ in terms of projective geometry, in order to study the dependence of  \eqref{Pseudospin def} on the choice of the basis appearing in \eqref{Components}.

We consider the complex projective space $\C P^1$ (the set of all complex lines in $\C^2$ passing through the origin), denoting by $[\psi_1, \psi_2]$ the line passing through the point $(\psi_1, \psi_2) \in \C^2 \setminus \set{0}$. In view of the identification $\C P^1 \simeq S^2$ via the stereographic projection from $S^2$ to the one-point compactification of the complex plane, the points $S = [0,1]$ and $N = [1,0]$ in $\C P^1$ are called South pole and North pole%
\footnote{Notice that, since the basis $\set{e_1, e_2}$ is orthonormal, the points $\set{N,S}$ are antipodal on $S^2$.}%
, respectively. Explicitly, stereographic projection is given by the map
\begin{equation} \label{stereo}
\C P^1 \to \C \cup \{\infty\}, \quad [\psi_1, \psi_2] \mapsto \psi_1 / \psi_2.
\end{equation}

Since every rank-$1$ orthogonal projector is identified with a point in $\C P^1$ (its range), the Datum \ref{Datum} yields a continuous map
\begin{equation} \label{Map G}
\begin{array}{rccccl}
 G: & S^1 &\longrightarrow  &  \Bi(\C^2) &  \longrightarrow &\C P^1  \\
    &  k  &\longmapsto      &    P(k)    &  \longmapsto     & \Ran P(k)=
[\psi_1(k), \psi_2(k)]  \\
\end{array}
\end{equation}
where, in the last equality, the choice of a basis in $\C^2$ is understood to write out the coordinates of $(\psi_1(k),\psi_2(k))$. With this identification, Assumption \ref{Assumption pseudospin} is equivalent to the condition
\begin{equation}\label{Assumption reinterpreted}
\mbox{the range of the map $G$ does not contain the points $N$ and $S$ in } \C P^1.
\end{equation}
Thus, in view of the previous Assumption, the range of $G$ is contained in a tubular neighbourhood of the equator $S^1_{\rm eq}$ of $S^2 \simeq \C P^1$, and therefore the degree of $G$ is well-defined. The latter is, up to a sign, the pseudospin winding number, \ie
\begin{equation} \label{nw=degG}
n\sub{w} = - \deg G, \qquad G: S^1 \rightarrow \mathrm{Tub}\, S^1_{\rm eq} \subset S^2.
\end{equation}
Indeed, clearly $G(k) = [\psi_1(k), \psi_2(k)] = [\psi_1(k) / \psi_2(k), 1]$; in view of our convention \eqref{stereo}, we deduce that
\[ \deg G = \deg g^{-1} = - \deg g = - n\sub{w} \]
as claimed.

\begin{figure}[ht]
\centering
\xy
(-7,0); (7,0) **\dir{-} ?>*\dir{>};
(70,0); (40,0), {\ellipse(,.2){.}};
(70,0); (40,0), {\ellipse(,){-}};
(18,-20); (40,-20), {\ellipse(,.1){--}};
(18,20); (40,20), {\ellipse(,.1){--}};
(12,10); (68,-10) **[red]\crv{~*=<\jot>{.} (68,5)};
(12,-10); (68,10) **[red]\crv{~*=<\jot>{.} (12,5)};
(12,10); (40,10), {\ellipse(,.2):a(220),^,:a(400){-}};
(12,-10); (40,-10), {\ellipse(,.2):a(220),^,:a(400){-}};
(-70,0); (-40,0), {\ellipse(,.2){-}};
(0,3)*{G}; (74,-.5)*{S^1_{\rm eq}}; (8,-15)*{G(S^1)};
(23,14)*{\mathrm{Tub}\, S^1_{\rm eq}};
(-40,-10)*{S^1};
(40,30)*{\bullet}; (43,32)*{N};
(40,-30)*{\bullet}; (43,-32)*{S};
(30,-10)*{\blu{\bullet}}; (33,-10)*{\blu{N'}};
(50,10)*{\blu{\bullet}}; (47,10)*{\blu{S'}};
\endxy
\caption{An example of map $G \colon S^1 \to \mathrm{Tub}\, S^1_{\rm eq} \subset S^2$.}
\label{fig:mapG}
\end{figure}

This reinterpretation shows that $n\sub{w}$ is only a \emph{conditional} topological invariant, in the sense that it is invariant only under continuous deformation of the Datum \ref{Datum} \emph{preserving Assumption \eqref{Assumption reinterpreted}}. The integer $n\sub{w}$ cannot be invariant under arbitrary continuous deformations: Indeed, the map $G$ in \eqref{Map G} can be continuosly deformed to a map $G\sub{South}$ which is constantly equal to the South pole, and then to a map $G_*$ which is constantly equal to an arbitrary point of the equator. Since $\deg(G_*)=0$, we conclude that \emph{the pseudospin winding number is not invariant under arbitrary continuous deformations of the family of projectors} appearing in Datum \ref{Datum}, or equivalently is not invariant under arbitrary continuous deformations of the corresponding Hamiltonian. In other words, $n\sub{w}$ is a  conditional topological invariant.

Finally, we point out that in general $n\sub{w}$ is not independent of the choice of the basis appearing in \eqref{Components}, which corresponds to a choice of antipodal points $\set{N,S}$ in $S^2$. Indeed, in the example in  Figure \ref{fig:mapG} one sees that $n\sub{w} = 2$ with respect to the choice $\set{N,S}$, while the choice $\set{N', S'}$ yields  $n'\sub{w} = 0$.

There exists a single case in which $n\sub{w}$ is independent of the basis (up to a sign), namely if
\begin{equation}\label{Equatorial range}
\mbox{the range of the map $G$ is contained in a maximum circle $E$ in $S^2$.}
\end{equation}
Indeed, in such a case one can identify $E$ with the equator, thus inducing a canonical choice of the poles $\set{N,S}$, up to reordering (\ie the only other possible choice is $N^\prime = S$ and $S^\prime =N$). This restrictive condition is indeed satisfied in the case of multilayer graphene, compare \eqref{Eigenfunction multilayer} and \eqref{Equatorial range} recalling that the condition $|\psi_1(k)| - |\psi_2(k)|= 0$ corresponds to being on the equator.

In summary, the PWN is well-defined and independent of the basis only under the restrictive assumptions \eqref{Assumption reinterpreted} and \eqref{Equatorial range}. These assumptions hold in the case of (multilayer) tight-binding graphene, but one cannot expect that they hold true in more general situations (\eg deformed graphene, topological insulators, \ldots). Indeed, the Haldane Hamiltonian $H\sub{Hal}(k)$ \cite{Haldane88}, which has been considered a paradigmatic model for many interesting effects in solid state physics, provides an example in which hypothesis \eqref{Equatorial range} does not hold. The Hamiltonian $H\sub{Hal}(k)$ is an effective $2 \times 2$ Hamiltonian, modelling a honeycomb crystal. It depends on several parameters: $t_1$ and $t_2$ which are hopping energies, $\phi$ which plays the role of an external magnetic flux, and $M$ which is an on-site energy. If $|t_2/t_1|<1/3$ and $M=\pm 3\sqrt{3} t_2 \sin(\phi)$, then the two bands of the Hamiltonian $H\sub{Hal}(k)$ touch at (at least) one point $k_0$ in the Brillouin zone. We checked that, for the values of the parameters $t_1 = 1$, $t_2 = (1/4) t_1$, $\phi=\pi/8$ and $M=3\sqrt{3} t_2 \sin(\phi)$ (in suitable units), the image of the map $G\sub{Hal}(k) := \Ran P\sub{Hal}(k)$, where $P\sub{Hal}(k)$ is the spectral projection on the upper band of $H\sub{Hal}(k)$ and $k$ varies in a circle around $k_0$, does \emph{not} lie on a maximum circle on the sphere $S^2 \simeq \C P^1$. Consequently, in the Haldane model, at least for these specific values of the parameters, the PWN is ill-defined.

%%%%%%%% SECTION 2.2.3 %%%%%%%%%%%%%%%%%%%%%%%%%%%%%%%%%%%%%%%%%%%%%%%%%%%%%%%%%%

\subsubsection{\textbf{Comparison of the two concepts.}}

In this Subsection, we are going to show that our eigenspace vorticity provides a more general and flexible definition of a topological invariant, which agrees with the PWN whenever the latter is well-defined, thus revealing its hidden geometric nature. In order to obtain, in the case $\Hi = \C^2$, a value of  $n\sub{v}$ which does not depend on the choice of the deformation, we will focus on a \virg{distinguished} class of deformations, namely those corresponding to weakly hemispherical maps (Definition \ref{Def hemispherical map}). 

Advantages and disadvantages of the two indices are easily noticed. As for the pseudospin winding number, its definition depends only on the \emph{undeformed} family of projectors, but it requires Assumption \ref{Assumption pseudospin} to be satisfied, with respect to a suitable basis of $\C^2$. Moreover, \eqref{Equatorial range} must also hold true for the definition of the PWN to be base-independent. Viceversa, the definition of the eigenspace vorticity  requires the construction of a family of \emph{deformed} projectors, but  it does not require any special assumption on the unperturbed family of projectors.

As the reader might expect, we can prove that the eigenspace vorticity and the pseudospin winding number coincide whenever both are defined. This holds true, in particular, in the case of monolayer and multilayer graphene.

Firstly, we give an alternative interpretation of the eigenspace vorticity as the degree of a certain map. This will make the comparison between the two indices more natural.

\begin{lemma} \label{nvdegree}
Let $\set{P(k)}_{k \in R \setminus \{k_0\}}$ be a family of rank-1 projectors as in Datum \ref{DatumProj}, with $\Hi = \C^2$, and let $\set{P^{\mu}(k)}_{(k,\mu) \in B}$ be a deformation of it, as described in Section \ref{sec:vorticity}. Let $\nv \in \Z$ be the eigenspace vorticity \eqref{n=ch1(L)} of the deformed family $\set{P^{\mu}(k)}_{{(k,\mu) \in B}}$. Define the map $\widetilde{G} \colon \mathcal{C} \to \C P^1$ by $\widetilde{G}(k, \mu) := \Ran P^{\mu}(k)$, for $(k,\mu) \in \mathcal{C}$. Then
\[ \nv = - \deg \widetilde{G}. \]
\end{lemma}
\begin{proof}
The integral formula for the degree of a smooth map $F \colon M \to N$ between manifolds of the same dimension \cite[Theorem 14.1.1]{Dubrovin} states that, if $\omega$ is a top-degree form on $N$, then
\begin{equation} \label{DegreeIntegralFormula}
\int_{M} F^* \omega = \deg F \int_{N} \omega.
\end{equation}

Let $\omega\sub{FS}$ be the \emph{Fubini-Study} $2$-form on $\C P^1$, defined as
\[ \omega\sub{FS}(\zeta) = \iu \overline{\partial} \partial \ln \left( 1 + |\zeta|^2 \right) \di \zeta \wedge \di \overline{\zeta} = \frac{\iu}{(1 + |\zeta|^2)^2} \di \zeta \wedge \di \overline{\zeta} \]
on the open subset $\C P^1 \setminus \set{S} = \set{ [\psi_1, \psi_2] \in \C P^1: \psi_1 \ne 0 } \simeq \C$ with complex coordinate $\zeta = \psi_2 / \psi_1$ (here $\partial = \partial / \partial \zeta$ and $\overline{\partial} = \partial / \partial \overline{\zeta}$). One easily checks that
\[ \frac{1}{2 \pi} \int_{\C P^1} \omega\sub{FS} = 1. \]

Moreover, it is also known \cite[Section 3.3.2]{Voisin} that $(1/2 \pi) \omega\sub{FS}$ is the first Chern class $\mathrm{Ch}_1(\mathcal{S})$ of the \emph{tautological bundle} $\mathcal{S}$ over $\C P^1$, whose fibre over the point representing the line $\ell \subset \C^2$ is the line $\ell$ itself. Instead, the bundle $\calL$ associated with the family $\set{P^{\mu}(k)}$ has the range of the projector as its fibre over $(k,\mu) \in \mathcal{C}$; this means by definition that it is the pullback via $\widetilde{G}$ of the tautological bundle. By naturality of the Chern classes, we deduce that the Berry curvature $2$-form $\omega$, defined as in \eqref{Berry curvature}, is given by%
\footnote{Notice that, without invoking the naturality of Chern classes, the equality $\omega = \widetilde{G}^* \omega\sub{FS}$ can also be explicitly checked by a long but straightforward computation.}%
\[ \omega = 2 \pi \, \mathrm{Ch}_1(\calL) = 2 \pi \, \mathrm{Ch}_1(\widetilde{G}^* \mathcal{S}) = 2 \pi \, \widetilde{G}^* \mathrm{Ch}_1(\mathcal{S}) = \widetilde{G}^* \omega\sub{FS}. \]
This fact, together with the formula \eqref{DegreeIntegralFormula}, yields to
\[ \nv = - \frac{1}{2 \pi} \int_{\mathcal{C}} \omega = - \frac{1}{2 \pi} \int_{\mathcal{C}} \widetilde{G}^* \omega\sub{FS} = - \deg \widetilde{G} \, \left(\frac{1}{2 \pi} \int_{\C P^1} \omega\sub{FS} \right) = - \deg \widetilde{G} \]
as claimed.
\end{proof}

We now proceed to the proof of the equality between the pseudospin winding number and the eigenspace vorticity, provided \eqref{Assumption reinterpreted} holds true.
We consider also cases, \eg perturbed graphene, in which a canonical orthonormal basis in $\C^2$ is provided by an unperturbed or reference Hamiltonian; so, condition \eqref{Equatorial range} is not assumed. However, whenever the Hamiltonian is such that  \eqref{Equatorial range}  holds true, then Step 1 in the following proof is redundant.

The class of deformations which we want to use to compute the eigenspace vorticity is defined as follows.

\begin{df}[Weakly hemispherical map]\label{Def hemispherical map}
A map $F \colon \mathcal{C} \to S^2$ is called \textbf{hemispherical} (with respect to a choice of an equator $S^1\sub{eq} \subset S^2$) if $F(S^1) \subset S^1\sub{eq}$ and $F^{-1}(S^1\sub{eq}) \subset S^1$, where $S^1 = \mathcal{C} \cap \set{\mu=0}$. Equivalently, $F$ is hemispherical if it maps the ``upper half'' of the cylinder (namely $\mathcal{C}_+ := \mathcal{C} \cap \set{\mu > 0}$) to the northern hemisphere $S^2_+$, the ``equator'' $S^1$ into $S^1\sub{eq}$, and the ``lower half'' of the cylinder (namely $\mathcal{C}_- := \mathcal{C} \cap \set{\mu < 0}$) to the southern hemisphere $S^2_-$.

A map $F \colon \mathcal{C} \to S^2$ is called \textbf{weakly hemisperical} if $F(S^1)$ is contained in a tubular neighbourhood $\mathrm{Tub} \, S^1\sub{eq}$ of the equator, and it is homotopic to a hemispherical map via the \emph{retraction along meridians}.
\end{df}

The retraction along meridians is defined as follows. Choose two open neighbourhoods $O_N$ and $O_S$ of the North and South pole, respectively, which do not intersect the tubular neighbourhood $\mathrm{Tub} \, S^1\sub{eq}$ containing $F(S^1)$. Let $\rho_t \colon S^2 \to S^2$ be the homotopy which as $t$ goes from $0$ to $1$ expands $O_N$ to the whole northern hemisphere and $O_S$ to the whole southern hemisphere, while keeping the equator $S^1\sub{eq}$ fixed (compare Figure \ref{fig:rhot}). Then let $F' := \rho_1 \circ F$; the maps $\rho_t \circ F$ give an homotopy between $F$ and $F'$. Then $F$ is weakly hemispherical if $F'$ is hemispherical.

\begin{figure}[ht]
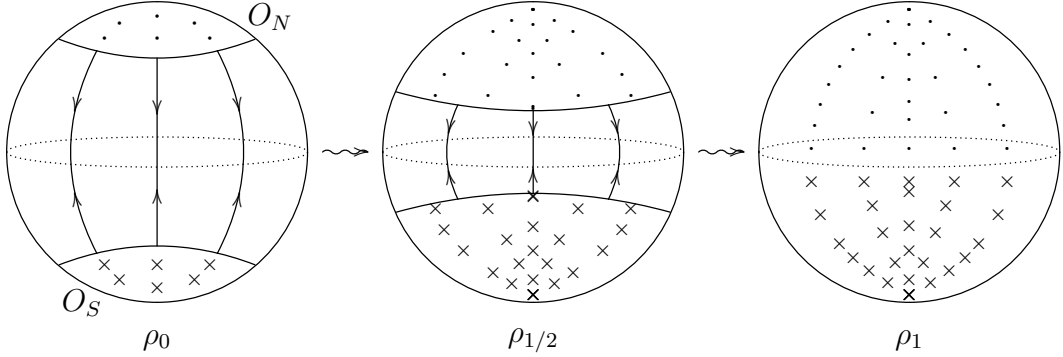

\centering
\xy
%%%%%% \rho_0
(-50,-25)*{\rho_0};
(-35,18)*{O_N}; (-60,-20)*{O_S};
(-70,0); (-50,0), {\ellipse(,){-}};
(-70,0); (-50,0), {\ellipse(,.1){.}};
(-63,15); (-37,15) **\crv{(-50,10)};
(-63,-15); (-37,-15) **\crv{(-50,-10)};
(-58,13.5); (-58,-13.5) **\crv{(-65,0)} ?(.7)*\dir{<} ?(.3)*\dir{>};
(-50,12.5); (-50,-12.5) **\dir{-} ?(.7)*\dir{<} ?(.3)*\dir{>};
(-42,13.5); (-42,-13.5) **\crv{(-35,0)} ?(.7)*\dir{<} ?(.3)*\dir{>};
(-50,15)*{\cdot}; (-50,18)*{\cdot}; (-57,15)*{\cdot};
(-55,17)*{\cdot}; (-43,15)*{\cdot}; (-45,17)*{\cdot};
(-50,-15)*{\scriptstyle{\times}}; (-50,-18)*{\scriptstyle{\times}}; (-57,-15)*{\scriptstyle{\times}}; (-55,-17)*{\scriptstyle{\times}}; (-43,-15)*{\scriptstyle{\times}}; (-45,-17)*{\scriptstyle{\times}};
(-28,0); (-22,0) **\dir{~} ?>*\dir{>};
%%%%%% \rho_{1/2}
(0,-25)*{\rho_{1/2}};
(-20,0); (0,0), {\ellipse(,){-}};
(-20,0); (0,0), {\ellipse(,.1){.}};
(-18.2,8); (18.2,8) **\crv{(0,3)};
(-18.2,-8); (18.2,-8) **\crv{(0,-3)};
(-10,6.3); (-10,-6.3) **\crv{(-13,0)} ?(.7)*\dir{<} ?(.3)*\dir{>};
(0,5.5); (0,-5.5) **\dir{-} ?(.7)*\dir{<} ?(.3)*\dir{>};
(10,6.3); (10,-6.3) **\crv{(13,0)} ?(.7)*\dir{<} ?(.3)*\dir{>};
(0,19); (-13,7.5) **\crv{~*=<10pt>{.} (-12,13)};
(0,19); (-6,7.5) **\crv{~*=<10pt>{.} (-4,10)};
(0,17); (0,6) **\crv{~*=<10pt>{.} (0,0)};
(0,19); (6,7.5) **\crv{~*=<10pt>{.} (4,10)};
(0,19); (13,7.5) **\crv{~*=<10pt>{.} (12,13)};
(0,-19); (-13,-7.5) **\crv{~*=<10pt>{\scriptstyle{\times}} (-12,-13)};
(0,-19); (-6,-7.5) **\crv{~*=<10pt>{\scriptstyle{\times}} (-4,-10)};
(0,-17); (0,-6) **\crv{~*=<10pt>{\scriptstyle{\times}} (0,0)};
(0,-19); (6,-7.5) **\crv{~*=<10pt>{\scriptstyle{\times}} (4,-10)};
(0,-19); (13,-7.5) **\crv{~*=<10pt>{\scriptstyle{\times}} (12,-13)};
(22,0); (28,0) **\dir{~} ?>*\dir{>};
%%%%%% \rho_1
(50,-25)*{\rho_1};
(30,0); (50,0), {\ellipse(,){-}};
(30,0); (50,0), {\ellipse(,.1){.}};
(50,19); (37,.5) **\crv{~*=<10pt>{.} (38,13)};
(50,19); (44,.5) **\crv{~*=<10pt>{.} (46,10)};
(50,17); (50,.5) **\crv{~*=<10pt>{.} (50,0)};
(50,19); (56,.5) **\crv{~*=<10pt>{.} (54,10)};
(50,19); (63,.5) **\crv{~*=<10pt>{.} (62,13)};
(50,-19); (37,-4) **\crv{~*=<10pt>{\scriptstyle{\times}} (38,-13)};
(50,-19); (44,-4) **\crv{~*=<10pt>{\scriptstyle{\times}} (46,-10)};
(50,-17); (50,-4) **\crv{~*=<10pt>{\scriptstyle{\times}} (50,0)};
(50,-19); (56,-4) **\crv{~*=<10pt>{\scriptstyle{\times}} (54,-10)};
(50,-19); (63,-4) **\crv{~*=<10pt>{\scriptstyle{\times}} (62,-13)};
\endxy
\caption{The retraction $\rho_t$ for $t = 0$, $t=1/2$ and $t=1$.}
\label{fig:rhot}
\end{figure}

\begin{prop}\label{Prop Equivalence ALTERNATIVE}
Let $\set{P(k)}_{k \in \partial U}$ be a family of projectors as in Datum \ref{Datum} and let $G \colon S^1 \to \C P^1$ be the corresponding map, \ie $G(k) = \Ran P(k)$ as in \eqref{Map G}. Suppose that Assumption \ref{Assumption pseudospin} (or equivalently the condition \eqref{Assumption reinterpreted} on $G$) holds. Let $\set{P^{\mu}(k)}_{(k,\mu) \in \mathcal{C}}$ be a deformed family of projectors, as in Section \ref{sec:vorticity}, and $\widetilde{G} \colon \mathcal{C} \to \C P^1$ be defined by $\widetilde{G}(k, \mu) := \Ran P^{\mu}(k)$. (Clearly $\widetilde G(k,0) = G(k)$ for $k \in S^1$.) Assume that $\widetilde{G}$ is weakly hemispherical. Then, up to a reordering of the basis involved in the definition of the pseudospin winding number, one has
\[ n\sub{w}(P) = n\sub{v}(P). \]
\end{prop}
\begin{proof}
The proof will consist in modifying suitably the functions $G$ and $\widetilde{G}$ (without leaving their respective homotopy classes), in order to compare their degrees. We divide the proof of this statement into a few steps.

{\it Step 1: choice of suitable maps $G$ and $\widetilde{G}$}. Consider the composition $\widetilde{G}' := \rho_1 \circ \widetilde{G}$, which by hypothesis can be assumed to be a hemispherical map (up to a reordering of the poles, \ie of the basis for the definition the PWN). As the degree of the map $\widetilde{G}$ depends only on the homotopy class of $\widetilde{G}$, we have that $\deg \widetilde{G} = \deg \widetilde{G}'$. The analogous statement holds also for $G$ and $G' := \rho_1 \circ G$. In the following, we drop the ``primes'' and assume that $G$ is such that $G(S^1) \subseteq S^1\sub{eq}$ and that $\widetilde{G}$ is hemispherical.

{\it Step 2: $\deg \widetilde{G} = \deg G$}. The map $\widetilde{G} \big|_{\mathcal{C}_+} \colon \mathcal{C}_+ \to S^2_+$ is a map between manifolds with boundary, mapping $\partial \mathcal{C}_+ = S^1$ to $\partial S^2_+ = S^1\sub{eq}$. By taking a regular value in $S^2_+$ to compute the degree of $\widetilde{G}$, we deduce from our hypotheses (Step 1) that the points in its preimage all lie in the upper half of the cylinder, so that $\deg \widetilde{G} = \deg \left( \widetilde{G} \big|_{\mathcal{C}_+} \right)$. On the other hand, the degree of a map between manifolds with boundary coincides with the degree of its restriction to the boundaries themselves \cite[Theorem 13.2.1]{Dubrovin}. We conclude that
\[ \deg \widetilde{G} = \deg \left( \widetilde{G} \big|_{\mathcal{C}_+} \right) = \deg \left( \widetilde{G} \big|_{\partial \mathcal{C}_+} \right) = \deg \left( \widetilde{G} \big|_{S^1} \right) = \deg G \]
as claimed.

{\it Step 3: conclusion}. To sum up, putting together Step 1 and Step 2, by \eqref{nw=degG} and Lemma \ref{nvdegree} we conclude that
\[ \nv = - \deg \widetilde{G} = - \deg G = n\sub{w} \]
and this ends the proof of the Proposition.
\end{proof}

Finally, we conclude that $n\sub{v}$ is intrinsically defined when $\Hi = \C^2$. Indeed,  the ambiguity in the choice of the deformation may be removed by choosing a deformation which corresponds to a weakly hemispherical map. Then the value of $n\sub{v}$ is independent of the choice of a specific deformation among this class, in view of the following Corollary of the proof.

\begin{crl} \label{Cor_Vorticity_hemispher}
Let  $\set{P^{\mu}(k)}_{(k,\mu) \in B}$ and $\{\widehat{P}^{\mu}(k) \}_{(k,\mu) \in B }$ be two deformations of the family $\set{P(k)}_{k \in {R} \setminus \set{k_0}}$, both corresponding to a weakly hemisperical map with respect to an equator $S^1\sub{eq} \subset S^2$. Then
\[
n\sub{v}(P) = n\sub{v}(\widehat{P}).
\]
\end{crl}

An example of hemispherical deformation $\set{P^{\mu}(k)}$ for a family of projectors $\set{P(k)}$ is provided by the canonical models \eqref{nProjectors} and \eqref{nmuProjectors} of last Subsection, as the reader may easily check. Since the corresponding Hamiltonian $H_m(q)$, as in \eqref{Hn} with $n=m$ and $e(q) = |q|^m$, is unitarily conjugated to tight-binding the Hamiltonian \eqref{Effective Ham multigraphene} of $m$-multilayer graphene, we get that both the eigenspace vorticity and the PWN of graphene are equal to $m$, in accordance with \eqref{Pseudospin graphene} and the above Proposition.

As an example of a weakly hemispherical deformation, we can instead consider the eigenprojectors of the Haldane Hamiltonian, for example for the values of the parameters cited above. By Lemma \ref{nvwelldefined}, if the values of the parameters $(M,\phi)$ of $H\sub{Hal}(k)$ are close to those corresponding to the tight-binding Hamiltonian of monolayer graphene, namely $(M,\phi) = (0,0)$, then its eigenprojectors will also have an eigenspace vorticity equal to $1$ in absolute value, in accordance with the numerical evaluation of Remark \ref{Numerical evaluation}.

%%%%%%%% SECTION 3 %%%%%%%%%%%%%%%%%%%%%%%%%%%%%%%%%%%%%%%%%%%%%%%%%%%%%%%%%%

\section{Universality of the canonical models} \label{sec:universality}

In this Section, we prove that the local models provided by the projectors \eqref{nProjectors}, together with their deformed version \eqref{nmuProjectors}, are \emph{universal}, meaning the following:  
if the ``outer'' eigenspace vorticity of the deformed family of projectors $\set{P^{\mu}_s(k)}$ is $N \in \Z$ and if we set $\nv := sN$, with $s \in \set{+,-}$, then the $\nv$-canonical projections $\set{P^{\mu}_{\nv,s}(k)}$ will provide an extension of the family, initially defined in $B$, to the whole $\widehat{B}$ (the notation is the same of Section \ref{sec:vorticity}).

More precisely, we will prove the following statement.

\begin{theorem} \label{Th:criterion}
Let $\calL_s$ be the complex line bundle over $B$ defined as in \eqref{Ls}. Let $N \in \Z$ be its vorticity around $U$, defined as in \eqref{n=ch1(L)}, and set $\nv := sN$ (recall that $s \in \set{+,-}$). Then there exists a complex line bundle $\widehat{\calL}_{s}$ over $\widehat{B}$ such that
\[ \widehat{\calL}_{s} \big|_{B} \simeq \calL_{s} \quad \text{and} \quad \widehat{\calL}_{s} \big|_{\dot{C}} \simeq \PB_{\nv, s}, \]
where $\PB_{\nv, s}$ is the \Composite bundle corresponding to the family of projectors \eqref{nmuProjectors} for $n =n\sub{v}$.
\end{theorem}

The bundle $\widehat{\calL}_{s}$, appearing in the above statement, allows one to interpolate the ``external data'' (the bundle $\calL_{s}$) with the canonical model around the singular point (the \Composite bundle $\PB_{\nv, s}$). 

The geometric core of the proof of Theorem \ref{Th:criterion} lies in the following result.

\begin{lemma} \label{ch1isom}
Under the hypotheses of Theorem \ref{Th:criterion}, the restriction of the bundles $\calL_{s}$ and $\PB_{\nv, s}$ to the cylindrical surface $\mathcal{C}$ are isomorphic:
\[ \calL_{s} \big|_{\mathcal{C}} \simeq \PB_{\nv, s} \big|_{\mathcal{C}}. \]
\end{lemma}

\begin{proof}
The statement follows from the well-known facts that line bundles on any CW-complex $X$ are completely classified by their first Chern \emph{class}, living in $H^2(X;\Z)$ (see \cite{Steenrod}), and that when $X = \mathcal{C} \simeq S^2$ then $H^2(\mathcal{C};\Z) \simeq \Z$, the latter isomorphism being given by integration on $\mathcal{C}$ (or, more formally, by evaluation of singular $2$-cocycles on the fundamental class $[\mathcal{C}]$ in homology). As a result, one deduces that line bundles on $\mathcal{C}$ are classified by their first Chern \emph{number}. As $\calL_{s}$ and $\PB_{\nv, s}$ have the same Chern number,  equal to $-N$ , when restricted to the surface $\mathcal{C}$, they are isomorphic.
\end{proof}

\begin{proof}[Proof of Theorem \ref{Th:criterion}]
Vector bundles are given by ``gluing'' together (trivial) bundles defined on open sets covering the base space, so we can expect that we can glue $\calL_{s}$ and $\PB_{\nv, s}$ along $\mathcal{C}$, given that their restrictions on $\mathcal{C}$ are isomorphic (as was proved in Lemma \ref{ch1isom}). The only difficulty we have to overcome is that $\mathcal{C}$ is a \emph{closed} subset of $\widehat{B}$.

We argue as follows. Let $T$ be an open tubular neighbourhood of $\mathcal{C}$ in $\widehat{B}$, and let $\rho \colon T \to \mathcal{C}$ be a retraction of $T$ on $\mathcal{C}$. As $T \cap \dot{C}$ is a deformation retract of $\mathcal{C}$ via the map $\rho$, we may extend the definition of $\calL_{s}$ to $T \cap \dot{C}$ by letting
\[ \calL_{s} \big|_{T \cap \dot{C}} := \rho^* \left(\calL_{s} \big|_{\mathcal{C}} \right). \]
Similarly, we can extend $\PB_{\nv, s}$ outside $\dot{C}$ setting
\[ \PB_{\nv, s} \big|_{T \cap B} := \rho^* \left(\PB_{\nv, s} \big|_{\mathcal{C}} \right). \]
With these definitions, one has
\begin{subequations}
\begin{equation} \label{LT=LC}
\calL_{s} \big|_T \simeq \rho^* \left( \calL_{s} \big|_{\mathcal{C}} \right)
\end{equation}
and similarly
\begin{equation} \label{PT=PC}
\PB_{\nv, s} \big|_T \simeq \rho^* \left( \PB_{\nv, s} \big|_{\mathcal{C}} \right).
\end{equation}
\end{subequations}
In fact, let $\mathcal{V}$ denote either $\calL_{s}$ or $\PB_{\nv, s}$. It is known \cite[Theorem 14.6]{Milnor} that all complex line bundles admit a morphism of bundles to $\mathcal{U}_{U(1)} = \left(EU(1) \xrightarrow{\pi}
\C P^\infty\right)$, the tautological bundle%
\footnote{The \emph{infinite complex projective space} $\C P^\infty$ is defined as the inductive limit of the system of canonical inclusions $\C  P^N \hookrightarrow \C P^{N+1}$; it can be thought of as the space of all lines sitting in some ``infinite-dimensional ambient space'' $\C^\infty$. It admits a \emph{tautological line bundle} with total space
\[ E U(1) := \left\{ (\ell, v) \in \C P^\infty \times \C^{\infty} : v \in \ell \right\}, \]
whose fibre over the line $\ell \in \C P^\infty$ is the line $\ell$ itself, viewed as a copy of $\C$.} %
on $\C P^\infty$, and that their isomorphism classes are uniquely determined by the homotopy class of the map between the base spaces. Let
\[ \begin{aligned}
\xymatrix{
E\left(\mathcal{V} \big|_T \right) \ar[d] \ar[r] & EU(1) \ar[d] \\
T \ar[r]^{f_T} & \C P^\infty }
\end{aligned}
\quad \text{and} \quad
\begin{aligned}
\xymatrix{
E\left( \rho^* \left(\mathcal{V} \big|_{\mathcal{C}} \right) \right) \ar[d] \ar[r] & EU(1) \ar[d] \\
T \ar[r]^{f_{\mathcal{C}}} & \C P^\infty }
\end{aligned} \]
be those two morphisms of bundles just described; then \eqref{LT=LC} and \eqref{PT=PC} will hold as long as we prove that $f_T$ and $f_{\mathcal{C}}$ are homotopic. Now, the following diagram%
\footnote{The symbol $\approx$ denotes homotopy of maps.}%
\[ \xymatrix{
T \ar[d]_{\rho} \ar[r]^{f_\mathcal{C}} \ar@{.>}[dr]^{\mathrm{Id}_T} & \C P^\infty \ar@{}[dl]|(.7){\approx} \\
\mathcal{C} \ar@{^{(}->}[r]_{\iota} & T \ar[u]_{f_T} } \]
where $\iota \colon \mathcal{C} \hookrightarrow T$ denotes the inclusion map, is clearly commutative. By definition of deformation retract, the maps $\iota \circ \rho$ and $\mathrm{Id}_T$ are homotopic: hence
\[ f_{\mathcal{C}} = f_T \circ \iota \circ \rho \approx f_T \circ \mathrm{Id} = f_T \]
as was to prove.

By Lemma \ref{ch1isom} we have $\calL_{s} \big|_{\mathcal{C}} \simeq \PB_{\nv, s} \big|_{\mathcal{C}}$, and hence also
\[ \rho^* \left( \calL_{s} \big|_{\mathcal{C}} \right) \simeq \rho^* \left( \PB_{\nv, s} \big|_{\mathcal{C}} \right). \]

Equations \eqref{LT=LC} and \eqref{PT=PC} thus give
\[ \calL_{s} \big|_T \simeq \PB_{\nv, s} \big|_T. \]
Hence the two line bundles $\calL_{s}$ and $\PB_{\nv, s}$ are isomorphic on the \emph{open} set $T$, and this allows us to glue them.
\end{proof}

\begin{rmk}[Families of projectors with many singular points]
We conclude this Section with some observations regarding singular families of projectors with more than one singular point (but still a finite number of them). Denote by $S = \{K_1, \ldots, K_M\} \subset \T_2^*$ the set of singular points of the family of rank-$1$ projectors $ \T_2^* \ni k \mapsto P_s(k)$; also, pick pair-wise disjoint open balls $U_i$ around $K_i$ (in particular, each $U_i$ contains no singular point other than $K_i$). Smoothen $\set{P_s(k)}$ into $\set{P_s^{\mu}(k)}$ as explained in Section \ref{sec:vorticity}, and compute the $M$ vorticity integers $N_i \in \Z$ as in \eqref{n=ch1(L)}; this involves the choice of a smoothing parameter $\mu \in [-\mu_0, \mu_0]$, and for simplicity we choose the same sufficiently small $\mu_0 > 0$ for all singular points. Call $\calL_s$ the line bundle associated to such smoothed family of projectors; it is a vector bundle over $ \left( \T_2^* \times [-\mu_0, \mu_0] \right) \setminus \bigcup_{i=1}^{M} C_i$, where $C_i := U_i \times (-\mu_0, \mu_0)$. The above arguments now show that we can find a bundle $\widehat{\calL}_s$, defined on $\left(\T_2^* \times [-\mu_0, \mu_0] \right) \setminus \left(S \times \{\mu = 0\}\right)$, such that $\widehat{\calL}_s$ resticts to $\calL_s$ whenever the latter is defined, while it coincides with $\PB_{\nv^{(i)}, s}$ on $\left( U_i \times [-\mu_0, \mu_0] \right) \setminus \left\{ \left( k = K_i, \mu = 0 \right) \right\}$, where $\nv^{(i)} := s N_i$. Thus, deformations of families of rank-$1$ singular projectors with $M$ singular points are uniquely determined by $M$-tuples of integers, their local vorticities.
\end{rmk}

%%%%%%%% SECTION 4 %%%%%%%%%%%%%%%%%%%%%%%%%%%%%%%%%%%%%%%%%%%%%%%%%%%%%%%%%%

\section{Decrease of Wannier functions in graphene} \label{sec:decrease}

In this Section, we will use the $n$-canonical eigenvectors $\phi_{n, \pm}(q)$, that were explicitly computed via the canonical models presented in Section \ref{sec:n=n}, to extract the rate of decay of the Wannier functions of graphene. Actually, we can prove a more general result, under the following

\begin{assumption} \label{Assumption main}
We consider a periodic Schr\"odinger operator, \ie an operator in the form $H_{\Gamma} = - \Delta + V_{\Gamma}$, acting in $L^2(\R^2)$, where $V_{\Gamma}$ is real-valued, $\Delta$-bounded with relative bound zero, and periodic with respect to a lattice $\Gamma \subset \R^2$. For a Bloch band $E_s$ of $H_{\Gamma}$, we assume that:
\begin{enumerate}[(i)]
  \item \label{item_i} $E_s$ intersects the Bloch band $E_{s-1}$\footnote{All the following statements remain true, with only minor modifications in the proofs, if $E_{s-1}$ is replaced by $E_{s+1}$.} at finitely many points $\set{K_1, \ldots, K_M}$, \ie $E_s(k_0) = E_{s-1}(k_0) =: E_0$ for every $k_0 \in \set{K_1, \ldots, K_M}$;
  \item \label{item_ii} for every $k_0 \in \set{K_1, \ldots, K_M}$ there exist constants $r, v_+, v_- > 0$ and $m \in \set{1, 2}$, possibly depending on $k_0$, such that for $|q| < r$
    \begin{equation} \label{power intersection}
    \begin{aligned}
     E_s(k_0 + q) - E_0 & = v_+ |q|^{m} + \Or(|q|^{m+1}), \\
     E_{s-1}(k_0 + q) - E_0 & = - v_- |q|^{m} + \Or(|q|^{m+1});
    \end{aligned}
    \end{equation}
  \item \label{item_iii} the following \emph{semi-gap condition} holds true:
  $$
    \inf_{k \in \T_2^*} \set{| E_s(k) - E_i(k)| \, : \, i \neq s, \: i \neq s-1} =: g > 0;
  $$
  \item \label{item_iv} the dimension of the eigenspace corresponding to the eigenvalue $E_s(k)$  (resp. $E_{s-1}(k)$) is $1$ for every $k \in \T_2^* \setminus \set{K_1, \ldots, K_M}$;
 \item \label{item_v} there exists a deformation $H^{\mu}(k)$ of the fibre Hamiltonian $H(k)$, as in \eqref{H(k)}, such that, if $P_*(k)$ denotes the spectral projection of $H(k)$ onto the eigenvalues $\set{E_s(k), E_{s-1}(k)}$, then for all $k_0 \in \set{K_1, \ldots, K_M}$ the operator $P_*(k_0) H^{\mu}(k) P_*(k_0)$ is close, in the norm-resolvent sense, to%
\footnote{With abuse of notation, we denote
\[ P_*(k_0) H_m^{\mu}(k) P_*(k_0) := \sum_{a,b \in \set{s, s-1}} \ket{u_{a}(k_0)} \left(H_m^{\mu}(k)\right)_{a,b} \bra{u_{b}(k_0)} \]
with respect to an orthonormal basis $\set{u_{s}(k_0),u_{s-1}(k_0)}$ of $\Ran P_*(k_0)$.} %
$P_*(k_0) H_m^{\mu}(k) P_*(k_0)$, where $H_m^{\mu}(k)$ is as in \eqref{Hnmu} with $n=m$ and $e(q) = |q|^m$, for the same $m$ as in item \eqref{item_ii} and uniformly for $|k-k_0|<r$ and $\mu \in [-\mu_0, \mu_0]$.
\end{enumerate}
\end{assumption}

\noindent Condition \eqref{item_i} corresponds to considering a semimetallic solid. Condition \eqref{item_ii} characterises the local behaviour of the eigenvalues; while some of our results hold true for any $m \in \N^\times$, we need the assumption $m \in \set{1,2}$ in order to prove Theorem \ref{Asymp_wU-wcan}. Condition \eqref{item_iii} guarantees that the Bloch bands not involved in the relevant intersection do not interfere with the local structure of the eigenprojectors. Condition  \eqref{item_iv} excludes  permanent degeneracy of the eigenvalues. Lastly, condition \eqref{item_v} corresponds to the assumption that the tight-binding Hamiltonian $H_m^{\mu}(k)$ is an accurate approximation, in the norm-resolvent sense, of the original continuous Hamiltonian. Notice that we crucially assume $d=2$.

It is usually accepted that Assumption \ref{Assumption main} holds true for (monolayer and bilayer) graphene (with $M=2$), since conditions \eqref{item_i}, \eqref{item_ii} and \eqref{item_iii} can be explicitly checked to hold within the tight-binding approximation, and these conditions are expected to be stable under approximations. Recently, C. Fefferman and M. Weinstein \cite{FeffermanWeinstein_review} provided sufficient conditions on $V_{\Gamma}$ yielding (i) and (ii) with $m =1$. Hereafter, the indices $\set{s,s-1}$ will be replaced by $\set{+,-}$ to streamline the notation.

We recall that the Wannier function depends on a choice of the phase for the Bloch function $\psi_+$ (or, equivalently, its periodic part $u_+$) according to the following definition.

\begin{df}\label{Def Wannier}
The \textbf{Wannier function} $w_+ \in L^2(\R^2)$ corresponding to the Bloch function $\psi_+$ for the Bloch band $E_+$ is
\begin{equation}\label{Wannier}
w_+(x) := \frac{1}{|\B|^{\half}} \int_{\B} \di k \, \psi_+(k, x)= \frac{1}{|\B|^{\half}} \int_{\B} \di k \, \expo{\iu k \cdot x} u_+(k, [x])
\end{equation}
where $\B$ is the fundamental unit cell for $\Gamma^*$, and $[x] = x \bmod \Gamma$.
\end{df}

\noindent Thus $w_+$ is nothing but the Bloch-Floquet anti-transform of the Bloch function $\psi_+$ (see Section \ref{sec:BlochHamiltonians}). Recall also that the decay rate of the Wannier function as $|x| \to \infty$ is related to the regularity of the corresponding Bloch function, see \eg \cite[Section 2.2]{Kuchment_book} or \cite[Equation (2.5)]{Panati Pisante}. In particular, if the Bloch function can be chosen to be $C^\infty$-smooth, then the associated Wannier function decays at infinity faster than the inverse of any polynomial.  

Since the Wannier function is defined by integration over $\B$, the smoothness of $k \mapsto \psi_+(k, \cdot)$ can be analysed separately in different regions of the Brillouin zone. The problem is therefore reduced to a local analysis of the Bloch functions around the intersection points, as detailed in the next Subsection.

%%%%%%%% SECTION 4.1 %%%%%%%%%%%%%%%%%%%%%%%%%%%%%%%%%%%%%%%%%%%%%%%%%%%%%%%%%%

\subsection{Reduction to a local problem around the intersection points} \label{sec:reduction}

First, notice that the Bloch function relative to the Bloch band $E_{+}$ is unique, up to the choice of a $k$-dependent phase: we assume that such a choice has been performed%
\footnote{\label{fndecay} The decay rate at infinity of Wannier functions may {\it a priori} depend on such a choice, but this does not happen provided the change of phase is sufficiently smooth, as detailed in Remark \ref{AppA:invariance}.}%
, and denote the corresponding Bloch function by $\psi_+(k)$.
Now, let $U = U_{k_0} := \set{k \in \B : |k - k_0| < r}$ be the neighbourhood of the intersection point $k_0 \in \set{K_1, \ldots, K_M}$ for which expansion \eqref{power intersection} holds, and let $\widetilde{\chi}_U(k)$ be a smoothed characteristic function for $U$, namely a smooth function supported in $U$ which is identically $1$ on a smaller disc $D \subset U$ of radius $\rho < r$. We assume that such smoothed characteristic functions are radially symmetric, \ie $\widetilde{\chi}_U(k_0 + q) = \widetilde{\chi}(|q|)$. Then the Bloch function $\psi_+$ may be written as
\[ \psi_+(k) = \sum_{i=1}^{M} \widetilde{\chi}_{U_{K_i}}(k) \psi_+(k) + \left(1 - \sum_{i=1}^{M} \widetilde{\chi}_{U_{K_i}}(k) \right) \psi_+(k) . \]
The summands in the first term, call them $\psi_U(k) = \psi_{U_{K_i}}(k)$, contain all the information regarding the crossing of the two energy bands at the points in $\set{K_1, \ldots, K_M}$. On the other hand, since the Bloch bands intersect \emph{only} at these points (Assumption \ref{Assumption main}\eqref{item_i}), then the last term, call it $\widetilde{\psi}(k)$, can be assumed to be smooth. As the Wannier transform \eqref{Wannier} is linear, the Wannier function $w_+(x)$ corresponding to the Bloch function $\psi_+(k)$ via \eqref{Wannier} splits as
\[ w_+(x) = \sum_{i=1}^{M} w_{U_{K_i}}(x) + \widetilde{w}(x) \]
with an obvious meaning of the notation. In the next Subsection we will prove (see Theorem \ref{Asymp_wcan}) that each of the functions $w_{U_{K_i}}$ has a power-law decay at infinity. Moreover, since $k \mapsto \widetilde{\psi}(k)$ is smooth, we can make the reminder term $\widetilde{w}$ decay as fast as the inverse of a polynomial of arbitrary degree, by choosing $\widetilde{\chi}$ sufficiently regular. Consequently we get that the asymptotic behaviour of the true Wannier function $w_+$ is determined by that of the functions $w_{U_{K_i}}$.

In the following, we will concentrate on one intersection point $k_0 \in \set{K_1, \ldots, K_M}$, and calculate the rate of decay at infinity of the function $w_{U}$. Let us also notice that, for reasons similar to the above, the behaviour of $w_U$ at infinity does not depend on the particular choice of the cutoff function $\widetilde{\chi}_U$, provided it is sufficiently regular. Indeed, if $\widetilde{\chi}_U$ and $\widetilde{\chi}_U'$ are two different cutoffs, both satisfying the above conditions, then their difference $\widetilde{\chi}_U'' : = \widetilde{\chi}_U - \widetilde{\chi}_U'$ is supported away from the intersection point $k_0$. This means that the corresponding $\psi_U''(k) = \widetilde{\chi}_U''(k) \psi_+(k)$ is regular, and this in turn implies that $w_U''$ decays fast at infinity. For this reason, we are allowed to choose a simple form for the function $\widetilde{\chi}_U$: in particular, we will choose $\widetilde{\chi}(|q|)$ to be polynomial in $|q|$ (of sufficiently high degree $N$) for $|q| \in (\rho, r)$, \ie
\begin{equation} \label{polynomialchi}
\widetilde{\chi}(|q|) = \begin{cases} 1 & \text{if } 0 \le |q| \le \rho, \\ \displaystyle \sum_{i=0}^{N} \alpha_i \, |q|^i & \text{if } \rho < |q| < r, \\ 0 & \text{if } |q| \ge r. \end{cases}
\end{equation}
The coefficients $\alpha_i$ are chosen in order to guarantee that $\widetilde{\chi}$ is as smooth as required. In particular, having $\widetilde{\chi} \in C^p([0,+\infty))$ requires $N \ge 2p-1$.

%%%%%%%% SECTION 4.2 %%%%%%%%%%%%%%%%%%%%%%%%%%%%%%%%%%%%%%%%%%%%%%%%%%%%%%%%%%

\subsection{\texorpdfstring{Asymptotic decrease of the $n$-canonical Wannier function}{Asymptotic decrease of the n-canonical Wannier function}}

We proceed to determine the rate of decay of $w_U$, extracting it from the models which were illustrated in Section \ref{sec:models}. Clearly
\begin{equation} \label{wU}
w_{U}(x) = \frac{1}{|\B|^{\half}} \int_{\B} \di k \, \psi_{U}(k, x) = \frac{1}{|\B|^{\half}} \int_{U} \di k \, \expo{\iu k \cdot x} \widetilde{\chi}_U(k) u_+(k, [x]).
\end{equation}

We claim that its asymptotic behaviour is the same as that of the \emph{$n$-canonical Wannier function}
\begin{equation} \label{wcan}
 w\sub{can}(x) := \frac{1}{|\B|^{\half}} \int_{U} \di k \, \expo{\iu k \cdot x} \widetilde{\chi}_U(k) \phi_n(k, [x]) = \frac{\expo{\iu k_0 \cdot x}}{|\B|^{\half}} \int_{U} \di q \, \expo{\iu q \cdot x} \widetilde{\chi}(|q|) \phi_n(q, [x])
\end{equation}
where $n$ is the vorticity of the family of projectors  $\set{P^{\mu}_+(k) }_{k \in R \setminus \set{k_0}, {\mu \in [-\mu_0, \mu_0]}}$ corresponding to the deformed fibre Hamiltonian $H^{\mu}(k)$, as in Assumption \ref{Assumption main}\eqref{item_v} (here $R \subset \B$ is a contractible region containing $U$ but not intersecting any of the chosen balls centered at the other intersection points), and
\begin{equation} \label{phin}
 \phi_n(q,[x]) := \phi_{n,+}(q)_1 \, u_+(k_0, [x]) + \phi_{n,+}(q)_2 \, u_-(k_0, [x])
\end{equation}
with $\phi_{n,+}(q) = \left( \phi_{n,+}(q)_1, \phi_{n,+}(q)_2 \right)$ as in \eqref{n-canonical}. We will motivate why the vorticity of 
$\set{P^{\mu}_+(k)}$ is non-zero in Subsection \ref{sec:TrueVSStratified}. To prove this claim, we will establish the following two main results of this Section.

\begin{theorem} \label{Asymp_wcan}
Let $w\sub{can}$ be the $n$-canonical Wannier function defined as in \eqref{wcan}, with $n \ne 0$. Then there exist two positive constants $R, c > 0$ such that
\begin{equation} \label{RateOfDecaySpelledOut}
\left| w\sub{can}(x) \right| \le \frac{c}{|x|^{2}} \quad \text{if} \quad |x| \ge R.
\end{equation}
\end{theorem}

\noindent The proof of the previous Theorem is postponed to Subsection \ref{sec:proof_Asymp_wcan}.

While the proof of Theorem \ref{Asymp_wcan} does not require restrictions on the value of $m$ in Assumption \ref{Assumption main}\eqref{item_ii}, it is crucial that $m \in \set{1,2}$ for the techniques employed in our proof of Theorem \ref{Asymp_wU-wcan}. In the following we will denote by $|X|^\alpha$ the operator acting on a suitable domain in $L^2(\R^2)$ by $\left(|X|^{\alpha} w \right)(x) := |x|^\alpha w(x)$.

\begin{theorem} \label{Asymp_wU-wcan}
There exists a choice of the Bloch function $\psi_+$ such that the following holds: For $w_U$ and $w\sub{can}$ defined as in \eqref{wU} and \eqref{wcan}, respectively, and $m \in \set{1,2}$ as in Equation \eqref{power intersection}, one has
\[ |X|^s \left( w_U - w\sub{can} \right) \in L^2(\R^2) \]
for $s \ge 0$ depending on $m$ as follows:
\begin{itemize}
    \item if $m = 1$, for all $s < 2$;
    \item if $m = 2$, for all $s < 1$.
\end{itemize}
\end{theorem}

\noindent The proof of the previous Theorem is postponed to Subsection \ref{sec:proof_Asymp_wU-wcan}.

By combining the above two Theorems with the fact that the decay at infinity of the true Wannier function $w_+$ is equal to the one of $w_U$, as was shown in the previous Subsection, we deduce at once the following

\begin{theorem} \label{Asymp_graphene}
Consider an operator $H_{\Gamma} = - \Delta + V_{\Gamma}$ acting in $L^2(\R^2)$ and a Bloch band $E_s$ satisfying Assumption \ref{Assumption main}. Then there exists a choice of the Bloch function relative to the Bloch band $E_s$ such that the corresponding Wannier function $w_+ \in L^2(\R^2)$ satisfies
the following $L^2$-decay condition:
\begin{equation} \label{RateOfDecayL2}
|X|^{\alpha} w_+ \in L^2(\R^2) \quad  \mbox{for every } 0 \le \alpha < 1.
\end{equation}
\end{theorem}

\begin{proof}
Fix $0 \le \alpha < 1$. The $L^2$-norm of the function $|X|^{\alpha} w_+$ can be estimated by
\begin{align}
\left\| |X|^{\alpha} w_+ \right\|_{L^2(\R^2)}^2 & = \int_{\R^2} \di x \, |x|^{2\alpha} |w_+(x)|^2 \le \nonumber \\
& \le \int_{\R^2} \di x \, |x|^{2\alpha} |w_+(x) - w\sub{can}(x)|^2 + \int_{\R^2} \di x \, |x|^{2\alpha} |w\sub{can}(x)|^2. \label{true-can+can}
\end{align}

In order to give a bound to the first integral, we use the result of Theorem \ref{Asymp_wU-wcan}. Firstly we write
\begin{align*}
\int_{\R^2} \di x \, |x|^{2\alpha} |w_+(x) - w\sub{can}(x)|^2 & = \int_{D_1} \di x \, |x|^{2\alpha} |w_+(x) - w\sub{can}(x)|^2 + \\
& \quad + \int_{\R^2 \setminus D_1} \di x \, |x|^{2\alpha} |w_+(x) - w\sub{can}(x)|^2
\end{align*}
where $D_1 \subset \R^2$ is the ball of radius $1$ around the origin. The first term on the right-hand side of the above equality is finite, because the function $w_+ - w\sub{can}$ is in $L^2(\R^2)$ (hence {\it a fortiori} in $L^2(D_1)$) and $|x|^{2\alpha} \le 1$ for $x \in D_1$. For $x \in \R^2 \setminus D_1$ and $\alpha \in [0,
1)$, by Theorem \ref{Asymp_wU-wcan} we conclude that
\[ \int_{\R^2 \setminus D_1} \di x \, |x|^{2\alpha} |w_+(x) - w\sub{can}(x)|^2 \le \left\| |X|^{\alpha} \left( w_+ - w\sub{can} \right) \right\|_{L^2(\R^2)}^2 \]
is finite.

To show that also the second summand in \eqref{true-can+can} is finite, we use instead Theorem \ref{Asymp_wcan}. Write
\[ \int_{\R^2} \di x |x|^{2\alpha} \left|w\sub{can}(x)\right|^2  = \int_{D_R} \di x |x|^{2\alpha} \left|w\sub{can}(x)\right|^2 + \int_{\R^2 \setminus D_R} \di x |x|^{2\alpha} \left|w\sub{can}(x)\right|^2. \]
Again by the fact that the Wannier function $w\sub{can}$ is in $L^2(\R^2)$, it follows that the first integral on the right-hand side is finite. As for the second summand, we use the estimate provided by Equation \eqref{RateOfDecaySpelledOut}; thus we have
\begin{equation} \label{polarintegral}
\int_{\R^2 \setminus D_R} \di x  |x|^{2\alpha} \left| w\sub{can}(x)\right|^2 \le \mathrm{const} \cdot \int_{R}^{\infty} \di |x| \, |x| \, |x|^{2 \alpha} \, |x|^{-4}
\end{equation}
and this last integral is convergent if and only if $\alpha < 1$.
\end{proof}

%%%%%%%% SECTION 4.2.1 %%%%%%%%%%%%%%%%%%%%%%%%%%%%%%%%%%%%%%%%%%%%%%%%%%%%%%%%%%

\subsubsection{\textbf{Proof of Theorem \ref{Asymp_wcan}}} \label{sec:proof_Asymp_wcan}

We proceed to the proof of Theorem \ref{Asymp_wcan}, establishing the rate of decay at infinity of the $n$-canonical Wannier function $w\sub{can}$ (corresponding to the $n$-canonical eigenvector for an eigenvalue crossing). For later convenience, we will prove a slightly more general result, concerning the Wannier function associated to a Bloch function which is obtained from the $n$-canonical eigenvector by multiplication times $q_j$, $j \in \set{1,2}$.

\begin{prop} \label{wnm}
Define
\[ w_{n,p}(x) := \frac{\expo{\iu k_0 \cdot x}}{|\B|^{\half}} \int_{U} \di q \, \expo{\iu q \cdot x} q_j^{p} \widetilde{\chi}(|q|) \phi_n(q, [x]) \]
where $j \in \set{1,2}$, $p \in \set{0,1}$, and $\phi_n(q,[x])$ is as in \eqref{phin}. Then there exist two positive constants $R, c > 0$ such that
\[ \left| w_{n,p}(x) \right| \le \frac{c}{|x|^{p+2}} \quad \text{for } |x| \ge R. \]
\end{prop}

Notice that the exponent in the power-law asymptotics for $w_{n,p}$ is \emph{independent} on $n$ (but the prefactor $c$ will depend on it, as will be apparent from the proof). The statement of Theorem \ref{Asymp_wcan} regarding the decay rate of $w\sub{can}$ is a particular case of the above, namely when $p = 0$.

\begin{proof}
For notational simplicity, we set $j = 1$, the case $j=2$ being clearly analogous.  Without loss of generality, we also assume that $n > 0$. We choose Cartesian coordinates in $\mathbb{R}^2$ such that $x = (0, |x|)$, and consequently $q \cdot x = - 2 \pi |q| \, |x| \sin \ph$. Since the $n$-canonical eigenfunction is
\[ \phi_{n, +}(q) = \eu^{\iu n \ph/2} \begin{pmatrix} \cos(n \ph/2) \\ \sin(n \ph/2) \end{pmatrix}, \]
we can write
\begin{equation} \label{wcossin}
w_{n, p}(x) = \frac{\eu^{\iu k_0 \cdot x}}{|\B|^{\half}} \left[ w_{\cos,p}(x) \, u_{+}(k_0,[x]) + w_{\sin,p}(x) \, u_{-}(k_0,[x]) \right]
\end{equation}
where
\begin{align*}
w_{\cos,p}(x) & := \int_{0}^{r} \di |q| \, |q|^{p+1} \widetilde{\chi}(|q|) \int_{0}^{2 \pi} \di \ph \, \eu^{- \iu 2 \pi |q| \, |x| \sin \ph} \cos(\ph)^p \eu^{\iu n \ph/2} \cos\left(\frac{n}{2}\ph\right), \\
w_{\sin,p}(x) & := \int_{0}^{r} \di |q| \, |q|^{p+1} \widetilde{\chi}(|q|) \int_{0}^{2 \pi} \di \ph \, \eu^{- \iu 2 \pi |q| \, |x| \sin \ph} \cos(\ph)^p \eu^{\iu n \ph/2} \sin\left(\frac{n}{2}\ph\right).
\end{align*}
The function $x \mapsto u_{\pm}(k_0,[x])$ is $\Gamma$-periodic, and as a consequence of its definition \eqref{BlochFunction} it is in the Sobolev space $W^{2,2}(\T^2_{Y})$, hence continuous; consequently, it is a bounded function. Thus, the upper bound on $\left| w_{n, p} \right|$ is completely determined by that on the functions $w_{\cos,p}$ and $w_{\sin,p}$.

Notice that for $p \in \set{0,1}$
\begin{align*}
\cos(\ph)^p \eu^{\iu n\ph/2} \cos\left(\frac{n}{2}\ph\right) & = \frac{1}{4} \left( \eu^{\iu (n+p) \ph} + \eu^{\iu (n-p) \ph} + \eu^{\iu p \ph} + \eu^{- \iu p \ph} \right), \\
\cos(\ph)^p \eu^{\iu n\ph/2} \sin\left(\frac{n}{2}\ph\right) & = \frac{1}{4 \iu} \left( \eu^{\iu (n+p) \ph} + \eu^{\iu (n-p) \ph} - \eu^{\iu p \ph} - \eu^{- \iu p \ph} \right),
\end{align*}
so that we can write
\begin{align*}
w_{\cos,p}(x) & = \frac{1}{4} \left( I_{n+p,p}(|x|) + I_{n-p,p}(|x|) + I_{p,p}(|x|) + I_{-p,p}(|x|) \right), \\
w_{\sin,p}(x) & = \frac{1}{4 \iu} \left( I_{n+p,p}(|x|) + I_{n-p,p}(|x|) - I_{p,p}(|x|) - I_{-p,p}(|x|) \right),
\end{align*}
where
\begin{align*}
I_{\ell,p}(|x|) & := \int_{0}^{r} \di |q| \, |q|^{p+1} \widetilde{\chi}(|q|) \int_{0}^{2 \pi} \di \ph \, \eu^{\iu \left( \ell \ph - 2 \pi |q| \, |x| \sin \ph \right)} = \\
& = \frac{1}{(2 \pi |x|)^{p+2}} \int_{0}^{2 \pi r |x|} \di z \, z^{p+1} \widetilde{\chi}\left(\frac{z}{2 \pi |x|}\right) \int_{0}^{2 \pi} \di \ph \, \eu^{\iu \left( \ell \ph - z \sin \ph \right)},
\end{align*}
with the change of variables $z = 2 \pi |q| \, |x|$.

Now, by definition \cite{Luke}
\begin{equation} \label{BesselDefinition}
\frac{1}{2\pi} \int_{0}^{2 \pi} \di \ph \, \eu^{\iu \left( \ell \ph - z \sin \ph \right)} = \frac{1}{2\pi} \int_{0}^{2 \pi} \di \ph \, \cos \left( \ell \ph - z \sin \ph \right) =: J_{\ell}(z)
\end{equation}
is the Bessel function of order $\ell$: thus, the functions $w_{\cos,p}$ and $w_{\sin,p}$ are combinations of integrals of the Bessel functions, which explicitly look like
\[ I_{\ell,p}(|x|) = \frac{1}{(2 \pi)^{p+1}} \,  |x|^{-p-2} \int_{0}^{2 \pi r |x|} \di z \, z^{p+1} \widetilde{\chi}\left(\frac{z}{2 \pi |x|}\right) J_{\ell}(z). \]

In order to evaluate these integrals and establish their asymptotic properties, we split
\[ I_{\ell,p}(|x|) = I^{(1)}_{\ell,p}(|x|) + I^{(2)}_{\ell,p}(|x|), \]
where
\begin{align*}
I^{(1)}_{\ell,p}(|x|) & := \frac{1}{(2 \pi)^{p+1}} \,  |x|^{-p-2} \int_{0}^{2 \pi \rho |x|} \di z \, z^{p+1} J_{\ell}(z), \\
I^{(2)}_{\ell,p}(|x|) & := \frac{1}{(2 \pi)^{p+1}} \,  |x|^{-p-2} \int_{2 \pi \rho |x|}^{2 \pi r |x|} \di z \, z^{p+1}  \widetilde{\chi}\left(\frac{z}{2 \pi |x|}\right) J_{\ell}(z).
\end{align*}
Notice that the function $\widetilde{\chi}$ does not appear in the integral $I^{(1)}_{\ell, p}$, since it is constantly equal to $1$ for $0 \le |q| < \rho$ (compare \eqref{polynomialchi}).

We now use the fact \cite[Sec. 2.5, Eqn. (6)]{Luke} that for large real $t \to \infty$
\[ \int_{0}^{t} \di z \, z^{\mu} \, J_{\nu}(z) = \dfrac{2^\mu \Gamma\left(\frac{\nu + \mu + 1}{2}\right)}{\Gamma\left(\frac{\nu - \mu + 1}{2}\right)} - \left( \frac{2}{\pi t} \right)^{1/2} t^{\mu} h(t), \]
whenever $\Re(\nu + \mu)> -1$, where $h(t) = f(t) \cos \theta(t) + g(t) \sin \theta(t)$ with
\[ \theta(t) = t - \nu \frac{\pi}{2} + \frac{\pi}{4}, \quad f(t) = 1 + \Or(t^{-2}), \quad g(t) = \Or(t^{-1}). \]
This allows to immediately compute the asymptotic rate of $I^{(1)}_{\ell,p}(|x|)$:
\begin{align*}
I^{(1)}_{\ell,p}(|x|) & = \dfrac{2^{p+1} \Gamma\left(\frac{\ell + p + 2}{2}\right)}{\Gamma\left(\frac{\ell - p}{2}\right)} \, \frac{1}{(2 \pi)^{p+1}} \,  |x|^{-p-2} - \\
& \quad - \frac{1}{(2 \pi)^{p+1}} \,  |x|^{-p-2} \, \left( \frac{2}{\pi \cdot 2 \pi \rho |x|} \right)^{1/2} (2 \pi \rho |x|)^{p+1} h(2 \pi \rho |x|) = \\
& = \dfrac{\Gamma\left(\frac{\ell + p + 2}{2}\right)}{\Gamma\left(\frac{\ell - p}{2}\right) \pi^{p+1}} \, |x|^{-p-2} - \frac{\rho^{p+(1/2)}}{\pi} \, |x|^{-3/2} h(2 \pi \rho |x|).
\end{align*}

We now compute also the asymptotics of $I^{(2)}_{\ell,p}(|x|)$ for large $|x|$. We use the explicit polynomial form \eqref{polynomialchi} for the cutoff function $\widetilde{\chi}$ in the interval $(\rho,r)$, thus obtaining
\begin{align*}
I & ^{(2)}_{\ell,p}(|x|) = \frac{1}{(2 \pi)^{p+1}} \,  |x|^{-p-2} \left( \int_{0}^{2 \pi r |x|} - \int_{0}^{2 \pi \rho |x|} \right) \di z \, z^{p+1}  \left( \sum_{i=0}^{N} \alpha_i \frac{z^i}{(2 \pi |x|)^{i}} \right) J_{\ell}(z) = \\
& = 2 \pi \sum_{i=0}^{N} \frac{\alpha_i}{(2 \pi |x|)^{p+i+2}} \left( \int_{0}^{2 \pi r |x|} - \int_{0}^{2 \pi \rho |x|} \right) \di z \, z^{p+i+1}  J_{\ell}(z) = \\
& = 2 \pi \sum_{i=0}^{N} \frac{\alpha_i}{(2 \pi |x|)^{p+i+2}} \left[ - \left( \frac{2}{\pi \cdot 2 \pi r |x|} \right)^{1/2} (2 \pi r |x|)^{p+i+1} h(2 \pi r |x|) - (r \leftrightarrow \rho) \right] = \\
& = \frac{r^{p+(1/2)}}{\pi} \left( - \sum_{i=0}^{N} \alpha_i \, r^i \right) |x|^{-3/2} h(2 \pi r |x|) + \frac{\rho^{p+(1/2)}}{\pi} \left( \sum_{i=0}^{N} \alpha_i \, \rho^i \right) |x|^{-3/2} h(2 \pi \rho |x|) = \\
& = \frac{\rho^{p+(1/2)}}{\pi} |x|^{-3/2} h(2 \pi \rho |x|)
\end{align*}
where in the last equality we used the fact that $\widetilde{\chi}(\rho)=1$ and $\widetilde{\chi}(r)=0$.

From these computations, we deduce that%
\footnote{The prefactor can be computed using the factorial relation $\Gamma(z+1)=z\Gamma(z)$: one obtains
\[ \dfrac{\Gamma\left(\frac{\ell + p + 2}{2}\right)}{\Gamma\left(\frac{\ell - p}{2}\right) \pi^{p+1}} = \begin{cases} \dfrac{\ell}{2 \pi} & \text{if } p=0, \\[10pt] \dfrac{\ell^2-1}{4 \pi^2} & \text{if } p=1. \end{cases} \]}%
\[ I_{\ell,p}(|x|) = \dfrac{\Gamma\left(\frac{\ell + p + 2}{2}\right)}{\Gamma\left(\frac{\ell - p}{2}\right) \pi^{p+1}} \, |x|^{-p-2} + \Or(|x|^{-p-j}) \quad \text{for all } j \in \N^\times. \]

We conclude that the upper bound on the rate of decay at infinity of both $w_{\cos,p}$ and $w_{\sin,p}$ (and hence that of the Wannier functions $w_{n, p}$) is given by a multiple of $|x|^{-p-2}$, for both $p = 0$ and $p=1$ and independently of $n \in \Z$.
\end{proof}

Notice that with the above proof we actually have a stronger control on the asymptotic behaviour of the function $w_{n,p}$ than what stated in Proposition \ref{wnm}: Indeed, from \eqref{wcossin} we see that $w_{n,p}$ is a combination of functions which behave asymptotically as $|x|^{-p-2}$ (up to arbitrarily higher order terms), times a $\Gamma$-periodic function of $x$.

%%%%%%%% SECTION 4.3 %%%%%%%%%%%%%%%%%%%%%%%%%%%%%%%%%%%%%%%%%%%%%%%%%%%%%%%%%%

\subsection{Asymptotic decrease of the true Wannier function}

In this Subsection, we will prove Theorem \ref{Asymp_wU-wcan}. The importance of this result lies in the fact that we can deduce from it an upper bound on the decay rate at infinity of the true Wannier function in the continuous model of (monolayer and bilayer) graphene, as in Theorem \ref{Asymp_graphene}. Indeed, what we will show in the following is that all the singularity of the Bloch function $u_+(k) \in \Hf$ at the intersection point $k_0$ is encoded in its components along the vectors $u_+(k_0)$ and $u_-(k_0)$ in $\Hf$. This is essentially a consequence of Assumption \ref{Assumption main}\eqref{item_iii}, namely of the fact that all
Bloch bands not involved in the intersection at $k_0$ are well separated from the intersecting ones.

%%%%%%%% SECTION 4.3.1 %%%%%%%%%%%%%%%%%%%%%%%%%%%%%%%%%%%%%%%%%%%%%%%%%%%%%%%%%%

\subsubsection{\textbf{True Bloch bundle vs. \Composite Bloch bundle}} \label{sec:TrueVSStratified}

Let us briefly summarise the geometric results of Section \ref{sec:universality}, in order to specify them to the case under study; we will use all the notation of that Section. Assume that $r > 0$ is so small that the ball of radius $2r$ is all contained in $R$. Denote as above by $\set{P_+(k) := \left| u_+(k) \right\rangle \left\langle u_+(k) \right|}_{k \in R \setminus{k_0}}$ the family of projectors corresponding to the Bloch band $E_+$; we use it as our Datum \ref{DatumProj}, so we smoothen it via a parameter $\mu \in [-\mu_0, \mu_0]$ and calculate the vorticity $n \in \Z$ using the eigenprojectors of the deformation $H^{\mu}(k)$ appearing in Assumption \ref{Assumption main}\eqref{item_v}. Theorem \ref{Th:criterion} then establishes the existence of a vector bundle $\widehat{\calL}$ on $\widehat{B}$ such that:
\begin{itemize}
 \item outside of a cylinder $C'$ of radius $r + r' < 2r$ centered at the singular point $(k=k_0, \mu = 0)$, the bundle $\widehat{\calL}$ \emph{coincides} with the bundle $\calL_+$, which is associated to the smoothed family of projectors $\set{P^{\mu}_+(k)}$;
 \item inside a smaller pointed cylinder $\dot{C}_1$ of radius $r_1 := r - r' > 0$ centered at the singular point $(k=k_0, \mu = 0)$, the bundle $\widehat{\calL}$ \emph{coincides} with the $n$-canonical stratified Bloch bundle $\PB_n = \PB_{n,+}$;
 \item inside the tubular neighbourhood $T$ of width $r'$ of the cylindrical surface $\mathcal{C} = \partial C$, the bundle $\widehat{\calL}$ is constructed extending the isomorphism of Hermitian bundles
 \begin{equation} \label{isoS2} \calL_+ \big|_{\mathcal{C}} \simeq \PB_n \big|_{\mathcal{C}} \end{equation}
 (see Lemma \ref{ch1isom}).
\end{itemize}

On the other hand, as $\dot C$ is a deformation retract of $\mathcal{C}$, the isomorphism \eqref{isoS2} extends to an isomorphism
\begin{equation} \label{isodotC} \calL_+ \big|_{\dot C} \simeq \PB_n \big|_{\dot C} \end{equation}
which, together with the first items of the list above, allows us to conclude that \emph{the bundle $\widehat{\calL}$ constructed via Theorem \ref{Th:criterion} is isomorphic to the Bloch bundle $\calL_+$ as bundles on the whole $\widehat{B}$}.

We now want to translate the information contained in the latter isomorphism in terms of the associated families of projectors, and of the ``sections'' of these bundles (\ie Bloch functions). In order to do so, we first show that the vorticity $n$ of the family $\set{P_+(k)}$ equals $m \in \set{1,2}$ as in Assumption \ref{Assumption main}\eqref{item_ii}, so that in particular it is non-zero. Indeed, consider the deformation $H^{\mu}(k)$ of the fibre Hamiltonian $H(k) = (- \iu \nabla_y + k)^2 + V_{\Gamma}(y)$, as in Assumption \ref{Assumption main}\eqref{item_v}, and a circle $\Lambda_*$ in the complex plane enclosing only the eigenvalues $\set{E_+^{\mu}(k), E_-^{\mu}(k)}$ for $|k-k_0|$ and $\mu$ sufficiently small. Then the Riesz integral
\[ P_*^{\mu}(k) := \frac{\iu}{2\pi} \oint_{\Lambda_*} \di z \, \left( H^{\mu}(k) - z \1 \right)^{-1} \]
defines a smooth family of projectors over $C = U \times(-\mu_0, \mu_0)$ (see \cite[Prop. 2.1]{Panati Pisante} for a detailed proof), such that $P_*^{\mu = 0}(k) = P_*(k)$. By using again the Riesz formula with a different contour $\Lambda^{\mu}(k)$ enclosing only the eigenvalue $E_+^{\mu}(k)$ (compare the proof of Proposition \ref{DerivatePtrue} below), we can realise $P_+^{\mu}(k)$ as a subprojector of $P_*^{\mu}(k)$. Denoting $\Pi := P_*(k_0)$, we then have
\begin{align*}
\norm{\Pi P_+^{\mu}(k) \Pi & - \Pi P_{m,+}^{\mu}(k) \Pi}_{\mathcal{B}(\Hf)} = \norm{P_+^{\mu}(k) - P_{m,+}^{\mu}(k) }_{\mathcal{B}(\Ran \Pi)} = \\
& = \left\| \frac{\iu}{2\pi} \oint_{\Lambda^{\mu}(k)} \di z \, \left[ \left( H^{\mu}(k) - z \1 \right)^{-1} - \left( H_m^{\mu}(k) - z \1 \right)^{-1} \right] \right\|_{\mathcal{B}(\Ran \Pi)} = \\
& \le \frac{1}{2 \pi} \, |\Lambda^{\mu}(k)| \, \left\| \left( H^{\mu}(k) - z \1 \right)^{-1} - \left( H_m^{\mu}(k) - z \1 \right)^{-1} \right\|_{\mathcal{B}(\Ran \Pi)}
\end{align*}
where $H_m(k)$ is as in \eqref{Hnmu} (with $n=m$ and $e(q) = |q|^m$) and $P_{m,+}^{\mu}(q)$ is its eigenprojector. The norm on the right-hand side of the above inequality is uniformly bounded by Assumption \ref{Assumption main}\eqref{item_v}, say by $\delta > 0$, and the length of the circle $\Lambda^{\mu}(k)$ can be made to shrink as $|q|^m$ when $|q| \to 0$. Thus we deduce that
\begin{equation} \label{stima}
\norm{P_+^{\mu}(k) - P_{n,+}^{\mu}(k) }_{\mathcal{B}(\Ran \Pi)} \le \mathrm{const} |q|^m \delta,
\end{equation}
and the right-hand side of the above inequality can be made smaller than $1$. By Lemma \ref{nvwelldefined}, we then have
\[ \nv(P_+) = \nv(P_{m,+}) = m \in \set{1,2} \]
by Assumption \ref{Assumption main}\eqref{item_ii}.

The above estimate gives the existence of a Kato-Nagy unitary $V^{\mu}(k)$, as in \eqref{W definition}, such that
\[ \Pi \, P^{\mu}_+(k) \, \Pi = V^{\mu}(k)\, P^{\mu}_{m,+}(k) \, V^{\mu}(k)^{-1}. \]
This can be restated in terms of the existence of a Bloch function $u_+^\mu$ (\ie an eigenfunction of $P^{\mu}_+$) such that
\begin{equation} \label{Piu=Uphi}
\Pi u_+^{\mu}(k) = V^{\mu}(k) \phi^{\mu}_n(k) \quad \text{for } (k,\mu) \in \dot{C},
\end{equation}
where
\[ \phi^{\mu}_n(q,[x]) := \phi^{\mu}_{m,+}(q)_1 \, u_+(k_0, [x]) + \phi^{\mu}_{m,+}(q)_2 \, u_-(k_0, [x]) \]
with $\phi^{\mu}_{m,+}(q) = \left( \phi^{\mu}_{m,+}(q)_1, \phi^{\mu}_{m,+}(q)_2 \right)$ as in \eqref{canonic_avoided}. Moreover, combining the definition \eqref{W definition} of $V^{\mu}(k)$ and the estimate \eqref{stima}, one easily checks that $V^{\mu}(k)$ actually extends smoothly to the whole cylinder $C$, in particular at $(k = k_0, \mu = 0)$.

We restrict our attention to the slice $\mu = 0$. Denote by $w\sub{eff}$ the Wannier function associated via \eqref{Wannier} to $U \ni k \mapsto \Pi \, u_+(k) =: u\sub{eff}(k)$. Arguing as in Section \ref{sec:reduction}, one deduces that the asymptotic decay of $w\sub{eff}$ is determined by the integration of $u\sub{eff}$ on $U$ in \eqref{Wannier}. Then the equality \eqref{Piu=Uphi} implies that
\begin{align*}
|\B|^{\half} w\sub{eff}(x) & \asymp \int_{U} \di q \, \expo{\iu q \cdot x} \, \widetilde{\chi}(|q|) \Pi u_+(k_0 + q, [x]) = \\
& = \sum_{b \in \set{+,-}} \left[ \int_{U} \di q \, \expo{\iu q \cdot x} \left(
\sum_{a \in \{1,2\}} V(k_0 + q)_{a,b} \, \widetilde{\chi}(|q|) \phi_{m,+}(q, [x])_a \right) \right] u_{b}(k_0,[x]) .
\end{align*}
As was already noticed, the functions $u_{\pm}(k_0,\cdot)$ do not contribute to the decay at infinity of $w\sub{eff}$. On the other hand, by Taylor expansion at $k_0$ we can write
\begin{equation} \label{Taylor}
V(k_0+q)_{a,b} = V(k_0)_{a,b} + \sum_{j=1}^{2} q_j \, \frac{\partial V}{\partial q_j}(k_0)_{a,b} + \sum_{j,\ell=1}^{2} q_j \, q_{\ell} R_{j,\ell}(q)
\end{equation}
where the remainder $R_{j,\ell}$ is $C^{\infty}$-smooth on $U$. Consequently, we get
\begin{equation} \label{remainder}
\begin{aligned}
\int_{U} \di q \, & \expo{\iu q \cdot x} V(k_0 + q)_{a,b} \, \widetilde{\chi}(|q|) \phi_{m,+}(q, [x])_a = V(k_0)_{a,b} \left( \int_{U} \di q \, \expo{\iu q \cdot x}  \, \widetilde{\chi}(|q|) \phi_{m,+}(q, [x])_a \right) + \\
& \qquad + \sum_{j = 1}^{2} \frac{\partial V}{\partial q_j}(k_0)_{a,b} \left( \int_{U} \di q \, \expo{\iu q \cdot x}  \, q_j \widetilde{\chi}(|q|) \phi_{m,+}(q, [x])_a \right) + \\
& \qquad + \sum_{j,\ell = 1}^{2} \int_{U} \di q \, \expo{\iu q \cdot x}  \, q_j q_l R_{j,\ell}(q) \widetilde{\chi}(|q|) \phi_{m,+}(q, [x])_a.
\end{aligned}
\end{equation}

The terms in brackets in the first and second summand have already been estimated in Proposition \ref{wnm}, and so are known to produce the rate of decay at infinity of $|x|^{-2}$ and $|x|^{-3}$, respectively. As for the third summand, we preliminarly notice that the function
\[ S(q) :=  q_j q_l R_{j,\ell}(q) \widetilde{\chi}(|q|) \phi_{m,+}(q, [x])_a \]
is in $W^{2,\infty}(U)$. Indeed, the map $q \mapsto R_{j,\ell}(q) \widetilde{\chi}(|q|)$ is smooth and bounded, while for all $r,s \in \set{1,2}$
\[ \frac{\partial^2}{\partial q_r \partial q_s} \left(q_j q_\ell \phi_{m,+}(q, [x])_a \right) \]
is in $L^{\infty}(U)$, since $\phi_{m,+}(q, [x])_a$ is homogeneous of order zero in $q$.

\noindent We can now proceed to an integration by parts: observe in fact that
\begin{align*}
- x_r \, x_s \, \int_{U} \di q \, \expo{\iu q \cdot x}  \, S(q) & = \int_{U} \di q \, \frac{\partial^2}{\partial q_r \partial q_s} \left(\expo{\iu q \cdot x}\right)  \, S(q) = \\
& = \int_{U} \di q \, \expo{\iu q \cdot x}  \, \frac{\partial^2 S}{\partial q_r \partial q_s}(q).
\end{align*}
The boundary terms vanish because $\widetilde{\chi}$ is zero on $\partial U$. By what we have shown above, we obtain that
\[ \left| x_r \, x_s \, \int_{U} \di q \, \expo{\iu q \cdot x}  \, S(q) \right| \le |U| \left\| \frac{\partial^2 S}{\partial q_r \partial q_s} \right\|_{\infty} < + \infty, \]
so that the third summand in \eqref{remainder} decays faster than $|x|^{-2}$ at infinity. We conclude that the asymptotic behaviour of $w\sub{eff}$ is that of $|x|^{-2}$.

\begin{rmk}[Invariance of the decay rate of Wannier functions] \label{AppA:invariance}
Given Theorem \ref{Asymp_wU-wcan}, whose proof will be completed in the next Subsection, the above argument shows also that the decay rate of the Wannier function $w_+$ is not affected by a change of phase in the corresponding Bloch function $\psi_+$, provided the phase is at least of class $C^2$. Indeed, we already know that the decay rate of $w_+$ depends only on the local behaviour of $\psi_+$ around the intersection point. The exchange of $\psi_+(k)$ with $\expo{\iu \theta(k)} \psi_+(k)$, for $k \in U$, is then equivalent to the exchange of $\phi_{m,+}(k)$ with $\expo{\iu \theta(k)} \phi_{m,+}(k)$. This exchange is implemented by the action of the unitary diagonal operator $V(k) = \expo{\iu \theta(k)} \1 \in U(2)$: if the dependence of $\theta(k)$ on $k \in U$ is $C^2$, then it is possible to perform a Taylor expansion as in \eqref{Taylor}, so that the above argument applies.
\end{rmk}

%%%%%%%% SECTION 4.3.2 %%%%%%%%%%%%%%%%%%%%%%%%%%%%%%%%%%%%%%%%%%%%%%%%%%%%%%%%%%

\subsubsection{\textbf{Proof of Theorem \ref{Asymp_wU-wcan}}} \label{sec:proof_Asymp_wU-wcan}

Finally, we proceed to show that the decay of $w\sub{eff}$, the Wannier function corresponding to $\Pi u_+(k)$, and of $w_U$, the Wannier function corresponding to the restriction to $U$ of $u_+(k)$, are the same. This is achieved by showing that their difference, which is the Wannier function associated to $(\1 - \Pi) u_+(k) =: u\sub{rem}(k)$, decays sufficiently fast at infinity (\ie at least faster than $|x|^{-2}$); this in turn will be proved by showing that $u\sub{rem}(k)$ is sufficiently smooth, say of class $W^{s,2}$ for some positive $s$. Indeed, in view of the results relating regularity of Bloch functions and asymptotic properties of the corresponding Wannier functions (compare \cite[Equation (2.5)]{Panati Pisante}), we have that
\begin{equation}\label{Decay rate}
\text{if } u\sub{rem} \in W^{s,2}(U;\Hf) \text{ then } |X|^s \left( w_U - w\sub{eff} \right) \in L^2(\R^2).
\end{equation}

Before establishing the Sobolev regularity of $u\sub{rem}$, we need to prove some estimates on the derivatives of the projector $P_+(k)$, and correspondingly on those of the Bloch function $u_+(k)$.

\begin{prop} \label{DerivatePtrue}
Let $\set{P_+(k)}_{k \in \dot U}$ be the family of eigenprojectors for the Hamiltonian $H_\Gamma = - \Delta + V_\Gamma$ as in Assumption \ref{Assumption main}. Let $m \in \set{1,2}$ be as in \eqref{power intersection}. Then, for all choices of multi-indices $I \in \set{1,2}^{N}$, $N \in \N$, there exists a constant $C_N > 0$ such that
\begin{equation} \label{Ptrue_estimates}
 \left\| \partial^N_{I} P_+(k) \right\|_{\mathcal{B}(\Hf)} \le \frac{C_N}{|q|^{Nm}} \qquad \text{for all } k = k_0 + q, \quad 0 < |q| < r.
\end{equation}
\end{prop}

\begin{proof}
We will explicitly prove the validity of estimates of the form \eqref{Ptrue_estimates} for $N \le 2$, as these are the only cases which will be needed later in the proof of Theorem \ref{Asymp_wU-wcan}. Using similar techniques, one can prove the result for arbitrary $N$.

The projector $P_+(k)$ can be computed by means of the Riesz integral formula, namely
\[ P_+(k) = \frac{\iu}{2 \pi} \oint_{\Lambda(k)} \di z \, (H(k) - z \1)^{-1}, \]
where $H(k) = (- \iu \nabla_y + k)^2 + V_{\Gamma}$ is the fibre Hamiltonian \eqref{H(k)}, and $\Lambda(k)$ is a circle in the complex plane, enclosing only the eigenvalue $E_{+}(k)$. We choose $\Lambda(k)$ with center on the real axis, passing through $E_0 = E_+(k_0)$ and having diameter $d(q) = 2 v_+ |q|^{m}$, where $v_+$ and $m$ are the constants appearing in Assumption \ref{Assumption main}\eqref{item_ii}; in particular, the length of the circle $\Lambda(k)$ is proportional to $|q|^m$. Hereafter we assume, without loss of generality, that $E_0 =0$ and that $r$ is so small that $2 v_+ r^m < (g/3)$, where $g$ is as in Assumption \ref{Assumption main}\eqref{item_iii}.

Recall that $\kappa \mapsto H(\kappa)$, $\kappa \in \C^2$, defines an analytic family in the sense of Kato (see Section \ref{sec:BlochHamiltonians}). In particular, in view of \cite[Thm. VI.4]{Reed-SimonI} one can compute the derivatives of $H(k)$ and of its spectral projection using either weak or strong limits with the same result. As a consequence, one obtains
\begin{equation} \label{Riesz}
 \partial_j P_+(k) = - \frac{\iu}{2 \pi} \oint_{\Lambda(k)} \di z \, (H(k) - z \1)^{-1} \partial_j H(k) (H(k) - z \1)^{-1}.
\end{equation}
Notice that the dependence of the contour of integration on $k$ does not contribute to the above derivative, because the region contained between two close circles $\Lambda(k)$ and $\Lambda(k+h)$ does not contain any point in the spectrum of $H(k)$.

\begin{figure}[ht]
\centerline{
\begin{xy}
(0,-25); (0,25) **\dir{-} ?>*\dir{>};
(-30,0); (100,0) **\dir{-} ?>*\dir{>};
(0,0)*{\bullet}; (20,0)*{\bullet}, {\ellipse(,){-}};
(-3,-3)*{E_0}; (19,-3)*{E_+(k)}; (42,13)*{\Lambda(k)};
(-19,0)*{\bullet}; (-19,-3)*{E_-(k)};
(90,0)*{\bullet}; (90,-3)*{E_{s+1}(k)};
\end{xy} }
\caption{The integration countour $\Lambda(k)$}
\label{fig:Lambda}
\end{figure}
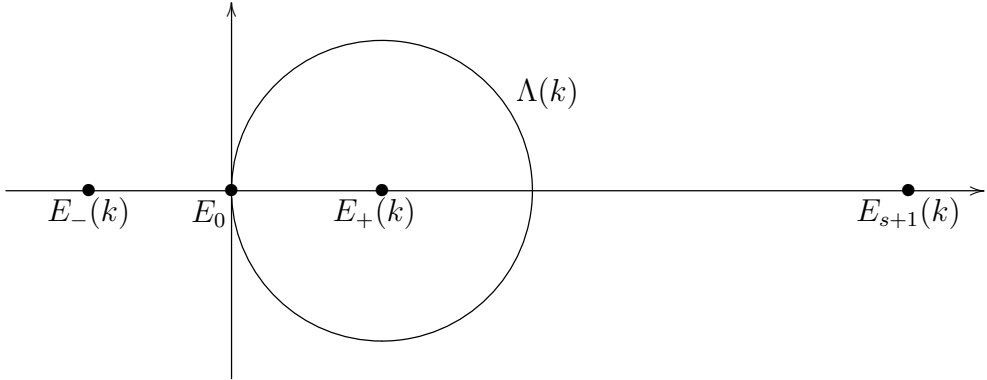

By the explicit expression of $H(k)$ given above, we deduce that
\[ \partial_j H(k) = 2 \left(- \iu \frac{\partial}{\partial y_j} + k_j\right) =: D_j(k). \]
Moreover, the resolvent of $H(k)$ is a bounded operator, whose norm equals
\[ \left\| (H(k) - z \1)^{-1} \right\|_{\mathcal{B}(\Hf)} = \frac{1}{\mathrm{dist} \left(z, \sigma(H(k))\right)} \]
where $\sigma(H(k))$ is the spectrum of the fibre Hamiltonian. As $z$ runs on the circle $\Lambda(k)$, the distance on the right-hand side of the above equality is minimal when $z$ is real, and the minimum is attained at the closest eigenvalue of $H(k)$, namely the selected Bloch band $E_+(k)$ or the eigenvalue $E_-(k)$. By Assumption \ref{Assumption main}\eqref{item_ii}, we conclude in both cases that
\begin{equation} \label{Resolvent}
\left\| (H(k) - z \1)^{-1} \right\|_{\mathcal{B}(\Hf)} \le \frac{c}{|q|^{m}}
\end{equation}
for some constant $c > 0$. In order to go further in our estimate, we need the following result.

\begin{lemma} \label{HRes}
There exists a constant $C > 0$ such that
\[ \left\| \partial_j H(k) (H(k) - z \1)^{-1} \right\|_{\mathcal{B}(\Hf)} \le \frac{C}{|q|^{m}} \]
uniformly in $z \in \Lambda(k)$, for any $k \in \dot U$.
\end{lemma}

\begin{rmk} \label{domains}
Notice that the operator $\partial_j H(k) (H(k) - z \1)^{-1}$ is indeed a bounded operator on $\Hf$. This is because the range of the resolvent $(H(k) - z \1)^{-1}$ is the domain $\mathcal{D} = W^{2,2}(\T^2_Y)$ of $H(k)$, and $k \mapsto H(k)$ is strongly differentiable, so that $\partial_j H(k)$ is well-defined on $\mathcal{D}$.
\end{rmk}

\begin{proof}[Proof of Lemma \ref{HRes}]
We recall \cite[Chap. XII, Problem 11]{Reed-Simon} that, if $H_0$ is a self-adjoint operator and $V, W$ are symmetric, then
\begin{equation} \label{Katosmall}
W << H_0 \text{ and } V << H_0 \quad \Longrightarrow \quad W << H_0 + V,
\end{equation}
where the notation $B << A$ means that $B$ is $A$-bounded with relative bound zero. Since $\iu \partial/\partial y_j << - \Delta$ and $V_\Gamma << - \Delta$ by Assumption \ref{Assumption main}, by iterating \eqref{Katosmall} we obtain
\[ \iu \frac{\partial}{\partial y_j} << - \Delta + 2 \iu k \cdot \nabla_y + |k|^2 \1 + V_\Gamma = H(k), \]
so that $D_j(k) << H(k)$. By definition of relative boundedness, this means that for any $a > 0$ there exists $b > 0$ such that
\[ \| D_j(k) \psi \|^2 \le a^2 \| H(k) \psi \|^2 + b^2 \| \psi \|^2 \]
for any $\psi$ in the domain of $H(k)$.

Fix $a > 0$. Then, by \cite[Lemma 2.40 and Eqn. (2.110)]{Amrein}, if $z_0 = \iu b/a$ one has
\begin{align*}
\| D_j(k) & (H(k) - z \1)^{-1} \|_{\mathcal{B}(\Hf)} \le \left\| D_j(k) \left( H(k) - z_0 \1 \right)^{-1} \right\|_{\mathcal{B}(\Hf)} + \\
& + \left| z - z_0 \right| \left\| D_j(k) \left( H(k) - z_0 \1 \right)^{-1} \right\|_{\mathcal{B}(\Hf)} \left\| \left( H(k) - z \1 \right)^{-1} \right\|_{\mathcal{B}(\Hf)} \le \\
& \le a + \left| z - z_0 \right| \, a \, \frac{c}{|q|^{m}}
\end{align*}
where we have used also Equation \eqref{Resolvent}. As $z$ varies in the circle $\Lambda(k)$, whose radius is proportional to $|q|^{m} \ll 1$, we can estimate
\[ \left| z - z_0 \right| \le \frac{b}{a} + 1 \quad \text{for every } z \in \Lambda(k), \text{ uniformly in } k \in \dot U. \]
As a consequence, there exists a constant $C > 0$ such that
\[ \left\| D_j(k) (H(k) - z \1)^{-1} \right\|_{\mathcal{B}(\Hf)} \le a + c \, \frac{a+b}{|q|^{m}} \le \frac{C}{|q|^{m}}, \]
yielding the claim.
\end{proof}

Plugging the result of this Lemma and Equation \eqref{Resolvent} into \eqref{Riesz}, and taking into account the fact that the length of $\Lambda(k)$ is proportional to $|q|^{m}$, we obtain that
\[ \left\| \partial_j P_+(k) \right\|_{\mathcal{B}(\Hf)} \le \frac{C_1}{|q|^{m}}. \]

The second derivatives of $P_+(k)$ can be computed similarly, again by means of the Riesz formula \eqref{Riesz} and taking into account Remark \ref{domains}: we obtain
\begin{align*}
\partial^2_{j,\ell} \, P_+(k) & = \frac{\iu}{2 \pi} \oint_{\Lambda(k)} \di z \, \left[ (H(k) - z \1)^{-1} \left( \partial_j H(k) (H(k) - z \1)^{-1} \right) \cdot \right. \\
& \quad \cdot \left( \partial_\ell H(k) (H(k) - z \1)^{-1} \right) + (j \leftrightarrow \ell) +  \\
& \left. \quad - (H(k) - z \1)^{-1} \partial^{2}_{j,\ell} H(k) (H(k) - z \1)^{-1} \right].
\end{align*}
Notice that $\partial^{2}_{j,\ell} H(k) = \delta_{j,\ell} \1_{\Hf}$, so that the last term in the above sum is easily estimated by $|q|^{-m}$. Again by \eqref{Resolvent} and Lemma \ref{HRes} we obtain that the first two terms are bounded by a constant multiple of $|q|^{-2m}$. We conclude that
\[ \left\| \partial^2_{j,\ell} \, P_+(k) \right\|_{\mathcal{B}(\Hf)} \le \frac{C_2}{|q|^{2m}}, \]
uniformly in $k \in \dot U$.
\end{proof}

From the above estimates on the derivatives of $P_+(k)$, we can deduce analogous estimates for the derivatives of the Bloch function $u_+(k)$.

\begin{prop} \label{Derivateutrue}
Let $\set{P_+(k)}_{k \in \dot U}$ and $m \in \set{1,2}$ be as in the hypotheses of Proposition \ref{DerivatePtrue}. There exists a function $k \mapsto u_+(k) \in \Hf$ such that $P_+(k) u_+(k) = u_+(k)$ and, for all choices of multi-indices $I \in \set{1,2}^{N}$, $N \in \N$, there exists a constant $C'_N > 0$ such that
\begin{equation} \label{utrue_estimates}
 \left\| \partial^N_{I} u_+(k) \right\|_{\Hf} \le \frac{C'_N}{|q|^{Nm}} \qquad \text{for all } k = k_0 + q, \quad 0 < |q| < r.
\end{equation}
\end{prop}

\begin{proof}
Again we will prove this statement only for $N \le 2$, the general case being completely analogous. Fix any point $k_* \in \dot{U}$ and $\delta < 1$: then, by continuity of the map $k \mapsto P_+(k)$ away from $k_0$, there exists a neighbourhood $U_* \ni k_*$ such that $\norm{P_+(k) - P_+(k_*)}_{\mathcal{B}(\Hf)} \le \delta < 1$ for all $k \in U_*$. The Kato-Nagy formula \eqref{W definition} then provides a unitary operator $W(k)$ which intertwines $P_+(k_*)$ and $P_+(k)$, in the sense that $P_+(k) = W(k) P_+(k_*) W(k)^{-1}$. Choose any $u_* \in \Ran P_+(k_*)$ with $\norm{u_*}_{\Hf} = 1$. Then, by setting $u_+(k):= W(k)u_*$ one obtains a unit vector in $\Ran P_+(k)$, and moreover
\[ \left\| \partial^N_I u_+(k) \right\|_{\Hf} \le \left\| \partial^N_I W(k) \right\|_{\mathcal{B}(\Hf)}. \]
Hence, if we prove that the Kato-Nagy unitary $W$ satisfies estimates of the form \eqref{Ptrue_estimates}, we can deduce that also \eqref{utrue_estimates} holds.

Set $Q(k) := \left(P_+(k) - P_+(k_*) \right)^2$, and notice that by hypothesis $\|Q(k)\| \le \delta^2$. Recall that the Kato-Nagy unitary is given by the formula
\[ W(k) = G(Q(k)) \left( P_+(k) P_+(k_*) + (\1 - P_+(k))(\1 - P_+(k_*))\right), \]
where $G(Q(k))$ is the function $G(z) := (1-z)^{-1/2}$ evaluated on  the self-adjoint operator $Q(k)$ by functional calculus. The function $G$ admits a power series expansion
\[ G(z) = \sum_{n=0}^{\infty} g_n \, z^n \]
which is absolutely convergent and term-by-term differentiable for $|z| \le \delta < 1$ (compare \eqref{powerseries}).

Differentiating the above expression for $W(k)$ with the use of the Leibniz rule for bounded-operator-valued functions, we obtain
\begin{align*}
\partial_j W(k) & = G(Q(k)) \left( \partial_j P_+(k) \right) \left(  2 P_+(k_*) - \1 \right) + \\
& \quad + \left( \partial_j G(Q(k)) \right) \left( P_+(k) P_+(k_*) + (\1 - P_+(k))(\1 - P_+(k_*))\right).
\end{align*}
As for the first summand, it follows from the properties of functional calculus that the term $G(Q(k))$ is bounded in norm by $G(\delta^2)$, while $2 P_+(k_*) - \1$ has norm at most equal to $3$. In the second summand, instead, the operator $( P_+(k) P_+(k_*) +$ $(\1 - P_+(k))(\1 - P_+(k_*)))$ is clearly uniformly bounded in $U_*$, while the norm of $\partial_j G(Q(k))$ can be estimated by
\[ \left\| \partial_j G(Q(k)) \right\| \le \left(\sum_{n = 0}^{\infty} n |g_n| \delta^{2(n-1)} \right) \left\|\partial_j Q(k) \right\|. \]
The term in brackets is finite because also the power series for the derivative of $G(z)$ is absolutely convergent. Moreover, we have that
\[ \partial_j Q(k) = \left(P_+(k) - P_+(k_*) \right) \left( \partial_j P_+(k) \right) + \left( \partial_j P_+(k) \right) \left(P_+(k) - P_+(k_*) \right) \]
from which we deduce that
\[ \left\| \partial_j Q(k) \right\| \le 2 \delta \left\| \partial_j P_+(k) \right\|. \]

In conclusion, using \eqref{Ptrue_estimates} we obtain that
\[ \left\| \partial_j W(k) \right\| \le \mathrm{const} \cdot \left\| \partial_j P_+(k) \right\| \le \frac{C_1'}{|q|^m}. \]

The second derivatives of $W(k)$ can be treated similarly, by using again Leibniz rule. One has
\begin{align*}
\partial^2_{j,\ell} W(k) & = G(Q(k)) \left( \partial^2_{j,\ell} P_+(k) \right) \left(  2 P_+(k_*) - \1 \right) + \\
& \quad + \left( \partial_j G(Q(k)) \right) \left( \partial_{\ell} P_+(k) \right) \left(  2 P_+(k_*) - \1 \right) + ( j \leftrightarrow \ell ) + \\
& \quad + \left( \partial^2_{j,\ell} G(Q(k)) \right) \left( P_+(k) P_+(k_*) + (\1 - P_+(k))(\1 - P_+(k_*))\right).
\end{align*}
The terms on the first two lines can be estimated as was done above by a multiple of $|q|^{-2m}$. For the one on the third line, we notice that
\begin{align*}
\left\| \partial^2_{j,\ell} G(Q(k)) \right\| & \le \left( \sum_{n=0}^{\infty} n (n-1) |g_n| \delta^{2(n-2)} \right) \left\| \partial_j Q(k) \right\| \, \left\| \partial_\ell Q(k) \right\| + \\
& \quad + \left( \sum_{n=0}^{\infty} n |g_n| \delta^{2(n-1)} \right) \left\| \partial^2_{j,\ell} Q(k) \right\|
\end{align*}
where the series in brackets are finite, due to the absolute convergence of the power series for the derivatives of $G(z)$. From the fact that
\begin{align*}
\partial^2_{j,\ell} Q(k) & = \left(P_+(k) - P_+(k_*) \right) \left( \partial^2_{j,\ell} P_+(k) \right) + \left( \partial^2_{j,\ell} P_+(k) \right) \left(P_+(k) - P_+(k_*) \right) + \\
& \quad + \left( \partial_j P_+(k) \right) \left( \partial_\ell P_+(k) \right) + \left( \partial_\ell P_+(k) \right) \left( \partial_j P_+(k) \right)
\end{align*}
we derive that
\[ \left\| \partial^2_{j,\ell} Q(k) \right\| \le 2 \delta \left\| \partial^2_{j,\ell} P_+(k) \right\| + 2 \left\| \partial_j P_+(k) \right\| \, \left\| \partial_\ell P_+(k) \right\|. \]

Putting all the pieces together, we conclude that
\[ \left\| \partial^2_{j,\ell} W(k) \right\| \le \mathrm{const} \cdot \left\| \partial^2_{j,\ell} P_+(k) \right\| + \mathrm{const} \cdot \left\| \partial_j P_+(k) \right\| \, \left\| \partial_\ell P_+(k) \right\| \le \frac{C_2'}{|q|^{2m}} \]
as wanted.
\end{proof}

With the help of the estimates \eqref{utrue_estimates}, we can establish the Sobolev regularity of the function $u\sub{rem}(k)$.

\begin{prop}
Let $u_+$ be as constructed in Proposition \ref{Derivateutrue} and $m \in \set{1,2}$ as in \eqref{power intersection}. Define $u\sub{rem}(k):= (\1 - \Pi) u_+(k)$. Then
\begin{itemize}
 \item if $m = 1$, one has $u\sub{rem} \in W^{s,2}(U;\Hf)$ for all $s < 2$;
 \item if $m = 2$, one has $u\sub{rem} \in W^{s,2}(U;\Hf)$ for all $s < 1$.
\end{itemize}
\end{prop}

\begin{proof}
We begin with the following simple observation: as $P_+(k)$ is a subprojector of $P_*(k)$, we have that
\[ (\1 - \Pi) P_+(k) = (P_*(k) - \Pi) P_+(k) \]
and consequently, as $u_+(k)$ is an eigenvector for $P_+(k)$, that
\[ u\sub{rem}(k) = (P_*(k) - \Pi) u_+(k). \]

From this, we deduce that
\[ \partial_j \, u\sub{rem}(k) = \partial_j P_*(k) \, u_+(k) + (P_*(k) - \Pi) \, \partial_j \, u_+(k). \]
The first summand is bounded in norm, because $P_*(k)$ is smooth in $k$ and $u_+(k)$ has unit norm.
As for the second summand, we know by \eqref{utrue_estimates} that
\[ \left\| \partial_j \, u_+(k) \right\| \le \frac{C_1'}{|q|^m}. \]
Moreover, by the smoothness of the map $k \mapsto P_*(k)$ and the definition of $\Pi = P_*(k_0)$, we deduce the Lipschitz estimate
\begin{equation} \label{Lip}
\left\| P_*(k) - \Pi \right\| \le L |q|
\end{equation}
for some constant $L > 0$. In conclusion, we get that
\[ \left\| \partial_j \, u\sub{rem}(k) \right\| \le \mathrm{const} \cdot |q|^{-m+1}. \]

Thus, if $m=1$, then $\partial_j u\sub{rem}$ is bounded, and hence the function $u\sub{rem}$ is in $W^{1,\infty}(U;\Hf)$. Instead, if $m=2$, then we can deduce that $u\sub{rem}$ is in $W^{1,p}(U;\Hf)$ for all $p<2$. Denoting by $\set{F_{p,q}^s}$ the scale of Triebel-Lizorkin spaces (see \eg \cite{RunstSickel}) one has that $W^{1,p} = F_{p,p}^1 \subseteq F^1_{p,\infty}$ is continuously embedded in $F^s_{2,2} = W^{s,2}$ for $s = 1 - d (1/p - 1/2)$, in view of \cite[Theorem 2.2.3]{RunstSickel}.  Thus, up to a continuous embedding, $u\sub{rem}$ is in $W^{s,2}(U;\Hf)$ for every $s < 1$, yielding the claim for $m=2$.

As for the second derivative
\[ \partial^2_{j,\ell} \, u\sub{rem}(k) = \partial^2_{j,\ell} P_*(k) \, u_+(k) + \big\{\partial_j P_*(k) \, \partial_\ell u_+(k) + (j \leftrightarrow \ell)\big\} + (P_*(k) - \Pi) \partial^2_{j,\ell} u_+(k), \]
we get that the first term is again bounded; the terms in brackets can be estimated by \eqref{utrue_estimates} with a multiple of $|q|^{-m}$; and the last summand, again by \eqref{utrue_estimates} and the Lipschitz estimate \eqref{Lip}, is bounded in norm by a multiple of $|q|^{-2m+1}$. For $m=1$, these two powers of $|q|$ coincide, and we conclude that $u\sub{rem}$ is in $W^{2,p}(U;\Hf)$ for all $p<2$. Again by the interpolation methods of \cite[Theorem 2.2.3]{RunstSickel}, we deduce that up to a continuous embedding $u\sub{rem}$ is in $W^{s,2}(U;\Hf)$ for all $s<2$, when $m=1$.
\end{proof}

The above result allows us to finally conclude the proof of Theorem \ref{Asymp_wU-wcan}. Indeed, by the considerations at the beginning of this Subsection (see Equation \eqref{Decay rate}) we have that
\[ |X|^s \left( w_U - w\sub{eff} \right) \in L^2(\R^2) \quad \begin{cases} \text{if } m = 1, \text{ for all } 0 \le s < 2 \\ \text{if } m = 2, \text{ for all } 0 \le s < 1 \end{cases} \]
which is exactly the statement of Theorem \ref{Asymp_wU-wcan}.

\appendix

%%%%%%%% APPENDIX A %%%%%%%%%%%%%%%%%%%%%%%%%%%%%%%%%%%%%%%%%%%%%%%%%%%%%%%%%%

\section{Distributional Berry curvature for eigenvalue intersections} \label{AppB}

Is there a way to define the Berry curvature also for singular families of projectors, \ie in presence of an eigenvalue intersection? Strictly speaking, notions like ``connection'' and ``curvature'' make sense only in the case of smooth vector bundles, as they are defined through differential forms. The eigenspace  bundle for an eigenvalue intersection, on the other hand, is singular at the intersection point $q=0$ -- see for example \eqref{ConicalProj}. Thus, in order to define a curvature also in the latter case, we have to ``pay a toll'': this amounts to using differential forms intepreted in a \emph{distributional sense}.

We make this last statement more rigorous, at least for the canonical families of projectors presented in Section \ref{sec:models}. We want to recover the Berry curvature for eigenvalue intersections from its analogue for avoided crossings, defining for all test functions $f \in C^\infty_0(U)$
\begin{equation} \label{DistributionalLimit}
\omega_{n, \pm}[f] := \lim_{\mu \downarrow 0} \omega_{n, \pm}^\mu[f] + \lim_{\mu \uparrow 0} \omega_{n, \pm}^\mu[f],
\end{equation}
where $T[f]$ denotes the action of the distribution $T$ on the test function $f$.

The distributions on the right-hand side of \eqref{DistributionalLimit} are the ones obtained from the Berry curvature (compare Equation \eqref{omegan})
\[ \omega_{n, \pm}^\mu(q) = \pm \frac{n}{2} \left[ \partial_{|q|} \left( \frac{\mu}{\sqrt{|q|^2+\mu^2}} \right) \di |q| \wedge \di \ph - \partial_{\mu} \left( \frac{\mu}{\sqrt{|q|^2+\mu^2}} \right) \di \ph \wedge \di \mu \right]. \]
Explicitly, when these act on a $\mu$-independent test function, the term containing $\di \mu$ does not contribute, yielding to
\begin{equation} \label{omegadistribution}
\omega_{n, \pm}^\mu[f] = \pm \frac{n}{2} \int_{U} \partial_{|q|} \left( \frac{\mu}{\sqrt{|q|^2+\mu^2}} \right) f(|q|,\ph) \, \di |q| \wedge \di \ph.
\end{equation}
The right-hand side of this equality changes sign according to whether $\mu$ is positive or negative, but still the two limits for $\mu \downarrow 0$ and $\mu \uparrow 0$ in \eqref{DistributionalLimit} give the same result, because the change in sign is compensated by the different orientation of the ``top'' and ``bottom'' caps of the cylinder $\mathcal{C}$, that are approaching the plane $\mu = 0$ (compare Figure \ref{fig:cylinder}). Hence,
\[ \omega_{n, \pm}[f] = \lim_{\mu \downarrow 0} \pm n \int_{U} \partial_{|q|} \left( \frac{\mu}{\sqrt{|q|^2+\mu^2}} \right) f(|q|,\ph) \, \di |q| \wedge \di \ph. \]

We may assume that $f$ is a radial-symmetric test function%
\footnote{Indeed, if $\widetilde{f}$ is in $C^\infty_0(U)$ and
\[ T_{n, \pm}^\mu(|q|) := \pm n \, \partial_{|q|} \left( \frac{\mu}{\sqrt{|q|^2+\mu^2}} \right), \]
then we have
\[ \omega_{n, \pm}^\mu\left[\widetilde{f}\right] = \int_{U}  T_{n, \pm}^\mu(|q|) \widetilde{f}(|q|,\ph) \, \di |q| \wedge \di \ph = \int_{0}^{r} \di |q| \, T_{n, \pm}^\mu(|q|) \int_{0}^{2\pi} \di \ph \widetilde{f}(|q|,\ph). \]
Set
\[ f(|q|) := \int_{0}^{2\pi} \di \ph \widetilde{f}(|q|,\ph). \]
Then $f$ is a $C^\infty$ function, because $\widetilde{f}$ is smooth and integration is performed on the compact set $S^1$; moreover, the support of $f$ is contained in a ball of radius $\widetilde{r}$ around $q=0$, where
\[ \widetilde{r} := \max_{q \in \operatorname{supp} \widetilde{f}} |q|. \]
By definition $\widetilde{r} < r$, because $\widetilde{f}$ has compact support in $U$; hence also $f$ is compactly supported in the same ball.}%
, \ie $f = f(|q|)$. We thus have
\begin{equation} \label{ApproxId}
\frac{1}{2\pi} \, \omega_{n, \pm}^\mu[f] = \pm n \int_{0}^{r} \di |q| \, \partial_{|q|} \left( \frac{\mu}{\sqrt{|q|^2+\mu^2}} \right) f(|q|).
\end{equation}
Notice that
\begin{align*}
\partial_{|q|} \left( \frac{\mu}{\sqrt{|q|^2+\mu^2}} \right) & = - \frac{1}{\mu} \, \frac{|q|/\mu}{\left[ 1 + (|q|/\mu)^2\right]^{3/2}} = j_{\mu}(-|q|),
\end{align*}
where
\[ j_{\mu}(|q|) = \frac{1}{\mu} j \left( \frac{|q|}{\mu} \right), \quad j(|q|) := \frac{|q|}{\left(1+|q|^2\right)^{3/2}}. \]

Consider for a moment the variable $|q|$ as varying on the whole real axis, and define
\[ J(|q|) := - j(|q|) \chi_{(-\infty, 0]}(|q|) \]
where $\chi_{(-\infty, 0]}$ is the characteristic function of the negative axis. The function $J$ then satisfies
\[ J(|q|) \ge 0 \text{ for all } |q| \in \R, \quad \text{and} \quad \int_{-\infty}^{+\infty} \di |q| \, J(|q|) = 1. \]
These two properties allow one to construct the so-called \emph{approximate identities} (compare \cite[Theorem 1.18]{SteinWeiss}), namely the functions
\[ J_{\mu}(|q|) := \frac{1}{\mu} j \left( \frac{|q|}{\mu} \right), \]
which are known to satisfy
\[ \lim_{\mu \downarrow 0} (J_{\mu} * F)(t) = \lim_{\mu \downarrow 0} \int_{-\infty}^{+\infty} \di |q| \, J_{\mu}(t-|q|) F(|q|) = F(t) \]
for every compactly-supported function $F$ and all $t \in \R$. Applying this to the function $F(|q|) = f(|q|) \chi_{(0,+\infty)}(|q|)$ (\ie extending $f$ to zero for negative $|q|$), we obtain that
\begin{align*}
f(0) & = F(0) = \lim_{\mu \downarrow 0} \int_{-\infty}^{+\infty} \di |q| \, J_{\mu}(-|q|) F(|q|) = \\
& = \lim_{\mu \downarrow 0} \int_{-\infty}^{+\infty} \di |q| \, \left(- \frac{1}{\mu} j\left(-\frac{|q|}{\mu}\right) \chi_{(-\infty,0]}\left(-\frac{|q|}{\mu}\right) \right) \chi_{[0,+\infty)}(|q|) f(|q|) = \\
& = - \lim_{\mu \downarrow 0} \int_{0}^{+\infty} j_{\mu}(-|q|) f(|q|).
\end{align*}
In the last of the above equalities we have used the fact that as $\mu > 0$
\[ \chi_{(-\infty,0]}\left(-\frac{|q|}{\mu}\right) = \chi_{(-\infty,0]}(-|q|) = \chi_{[0,+\infty)}(|q|). \]

Now it suffices to observe that on the right-hand side of \eqref{ApproxId} we have the expression
\[ \int_{0}^{r} \di |q| \, \partial_{|q|} \left( \frac{\mu}{\sqrt{|q|^2+\mu^2}} \right) f(|q|) = \int_{0}^{+\infty} \di |q| \, j_{\mu}(-|q|) f(|q|) \]
because, as $f$ has compact support in $U$, integration on $[0,r]$ or on $[0,+\infty)$ in the $|q|$-variable yields the same result. When $\mu$ approaches $0$ from above, we thus obtain
\[ \frac{1}{2\pi} \omega_{n, \pm}[f] = \lim_{\mu \downarrow 0} \frac{1}{2\pi} \omega_{n, \pm}^\mu[f] = \mp n f(0). \]

In conclusion, as a distribution $\omega_{n, \pm}$ is a multiple of the Dirac delta ``function'':
\[ \frac{1}{2\pi} \omega_{n, \pm} = \mp n \delta_0. \]
In other words, the curvature of the (singular, hence ill-defined!) eigenspace ``bundle'' for an eigenvalue intersection is all concentrated at the intersection point $q = 0$. Moreover, we can define a Chern number for this eigenspace ``bundle'' $\PB_{n,\pm}^0$, with abuse of notation, by posing
\[ \ch_1\left(\PB_{n,\pm}^0\right) = \frac{1}{2\pi} \int_{U} \omega_{n, \pm} = \mp n \int_{U} \delta_0 = \mp n. \]

%%%%%%%%%%%%%%%% BIBLIOGRAPHY %%%%%%%%%%%%%%%%%

\bigskip \bigskip

{\footnotesize

(D. Monaco) \textsc{SISSA, Via Bonomea 265, 34136 Trieste, Italy}

\medskip

{\it E-mail address}: \href{mailto:dmonaco@sissa.it}{\texttt{dmonaco@sissa.it}}

\bigskip

(G. Panati) \textsc{Dipartimento di Matematica ``G. Castelnuovo'', ``La Sapienza'' Universit\`{a} di Roma, Piazzale A. Moro 2, 00185 Roma, Italy}

\medskip

{\it E-mail address}: \href{mailto:panati@mat.uniroma.it}{\texttt{panati@mat.uniroma.it}}
}
\end{document}